\documentclass[journal,12pt,onecolumn]{IEEEtran}

\usepackage{mathtools}
\usepackage{amssymb}
\usepackage{amsmath,color,amsthm}
\usepackage{breqn}
\usepackage{bm}
\usepackage{cite}

\usepackage{algpseudocode}
\usepackage{algorithm}
\algnewcommand\algorithmicforeach{\textbf{for each}}
\algdef{S}[FOR]{ForEach}[1]{\algorithmicforeach\ #1\ \algorithmicdo}
\usepackage{eqparbox}
\newdimen{\algindent}
\algnewcommand\LeftComment[2]{%
\hspace{#1\algindent}$\triangleright$ \eqparbox{COMMENT}{#2} \hfill %
}
\algnewcommand\LeftCommentNoTriangle[2]{%
\hspace{#1\algindent} \eqparbox{COMMENT}{#2} \hfill %
}
\usepackage{enumitem}
\usepackage{caption}
\usepackage{lipsum}
\usepackage{cuted}
\usepackage{enumerate}
\usepackage{booktabs} % For prettier tables
\usepackage{graphicx}
\usepackage{float}
\usepackage{mathrsfs}

\DeclarePairedDelimiter{\nint}\lceil\rceil

\usepackage{amsmath,amsthm,amssymb}
\DeclareMathOperator*{\argmax}{argmax}
\newcommand{\gd}[2]{\mathcal G ( #1, #2) }
\newtheorem{theorem}{Theorem}
\newtheorem{corollary}{Corollary}
\newtheorem{lemma}{Lemma}

\usepackage{wasysym} %this lets me make smiley faces :-)

\DeclarePairedDelimiterX\Basics[1](){ #1}

\begin{document}
	\title{Multiple Access in Dynamic Cell-Free Networks: Outage Performance and Deep Reinforcement Learning-Based Design} 
%	\author{
%		\IEEEauthorblockN{Author, {\it Member, IEEE},  
%			 Author, {\it Fellow, IEEE,}} Author, {\it Member, IEEE} } 
   
\author{Yasser Al-Eryani, Mohamed Akrout, and Ekram Hossain \thanks{Y. Al-Eryani and E. Hossain are with the Department of Electrical and Computer Engineering at the University of Manitoba, Canada (emails: aleryany@myumanitoba.ca, Ekram.Hossain@umanitoba.ca). Mohamed Akrout is affiliated with the Department of Computer Science at the University of Toronto, Canada (email: makrout@cs.toronto.edu). 
The work was supported by a Discovery Grant form the Natural Sciences and Engineering Research Council of Canada (NSERC).}}
 	\maketitle
	\begin{abstract}
 In future cell-free (or cell-less) wireless networks, a large number of devices in a geographical area will be served simultaneously in non-orthogonal multiple access scenarios by a large number of distributed access points (APs), which coordinate with a centralized processing pool. 
 For such a centralized cell-free network with static predefined beamforming design, we first derive a closed-form expression of the uplink  per-user probability of outage.
To significantly reduce the complexity of joint processing of users' signals in presence of a large number of devices and APs, we propose a novel dynamic cell-free network architecture. In this architecture, the distributed APs are partitioned (i.e. clustered) among a set of subgroups with each subgroup acting as a virtual AP equipped with a distributed antenna system (DAS). The conventional static cell-free network is a special case of this dynamic cell-free network when the cluster size is one. For this dynamic cell-free network, we propose a successive interference cancellation (SIC)-enabled signal detection method and an inter-user-interference (IUI)-aware DAS's receive diversity combining scheme. 
%Such a clustering is conducted based on the current channel state information (CSI) between APs and all users within the network coverage area.
We then formulate the general problem of clustering APs and designing the beamforming vectors with an objective to maximizing the sum rate or maximizing the minimum rate. To this end, we propose a hybrid deep reinforcement learning (DRL) model, namely, a deep deterministic policy gradient (DDPG)-deep double Q-network (DDQN) model, to solve the optimization problem for online implementation with low complexity.
%The proposed dynamic cell-free network with hybrid DRL model offers a significant complexity-delay reduction with performance loss satisfactorily compensated by the utilization of SIC detection at the receiver side. 
The DRL model for sum-rate optimization significantly outperforms that for maximizing the minimum rate in terms of average  per-user rate performance. Also, in our system setting, the proposed DDPG-DDQN scheme is found to achieve around $78\%$ of the rate achievable through an exhaustive search-based design. 
%However, complexity and processing time was decreased significantly by the adoption of such DRL model.  
\end{abstract}
 
%==================================================================================
\section{Introduction}
%====================================================================================

\subsection{Background and Related Work}
Network densification through the deployment of more access points(APs)/base stations (BSs) per unit area  and connecting them in a multi-tier network architecture is a natural way to improve cellular network capacity~\cite{HetNet_1}. However, network densification gives rise to critical issues such as increased  signal interference (both inter-cell interference [ICI] and inter-user interference [IUI]), as well as increased network complexity and capital and operating expenditures.
Many interference-tolerant and low-cost technologies including massive multiple-input multiple-output antenna (mMIMO), transmit/receive diversity (TRD), and distributed antenna system (DAS)~\cite{massiveMIMO_1,R_Diversity_1,R_Diversity_2,Distributed_Ant_1} have therefore been developed. 
The performance gains in these technologies are mainly achieved by increasing the received power of the desired signal without any coordination among  interference sources in other cells~\cite{UltraDense_1}. The interference can however be minimized by coordinating wireless transmitters and/or receivers (either APs or user equipment [UE]).  For instance, coordinated multi-point (CoMP) technique allows adjacent APs to coordinate their transmission and reception to serve cell-edge UEs. This will turn the strongest interference signals into a desired ones at the cost of decreased spectral efficiency \cite{CoMP6146494}.

Recently, the concept of cooperation among distributed APs has been generalized in which all APs serve all users  (i.e. both cell-edge and cell-center users) within the network coverage area. 
This concept can be implemented through several techniques such as cloud radio access network (C-RAN)  \cite{CRAN,CRAN_1}, cell-free networks (or cell-less networks) \cite{Cell_Less_1,Cell_Less_2,Cell_Less_3,Cell_Less_4} and Generalized CoMP (GCoMP) networks~\cite{DOMA,GCoMP}. 
In a C-RAN, a group of distributed remote radio heads (RRHs) that are geographically distributed in the network coverage area are used to perform all radio functionalities such as high frequency amplification, frequency up/down conversion and A/D and D/A conversion. 
The users' baseband signals coming from all RRHs are then sent to a virtual baseband unit (BBU) pool that is located in the network cloud and is responsible to perform all baseband functionalities such as signal detection, etc. 
%C-RAN allows the use of different MA (different wireless technologies) within the network coverage area in a hierarchical setup (similar to HetNets) and was potentially developed for conventional multi-tier cellular networks. 
In a cell-free network architecture, a group of distributed APs cooperate to serve all active users within the network coverage area simultaneously using the same frequency-time resources. This makes the entire network act as a massive distributed antenna system (DAS) that covers the entire network coverage area \cite{DAS_1}. 
This network architecture however will require a huge processing capability of the BBU pool (or the central processing unit [CPU]), especially with the massive increase on the number of users per unit area in beyond 5G/6G networks \cite{6G8412482,DBLP:journals/corr/abs-1901-07106}. 

Several studies on cell-free networks have investigated performance metrics such as per-user transmission rate \cite{Cell_Less_3},  per-user packet delay \cite{Cell_Less_Delay_1}, implementation issues such as pilot contamination \cite{Cell_Less_1}, network backhauling/fronthauling \cite{Fronthaul_1,Fronthaul_2}, beamforming techniques \cite{Cell_Less_Beamforming_1}. 
In \cite{Cell_Less_1}, it was found that channel estimation error resulting from non-orthogonal pilot sequences tends to significantly decrease the network performance.  
The authors in \cite{Fronthaul_2} evaluated the per user performance under limited fronthaul link capacity and showed that the network performance degrades as the number of active users increases under limited fronthaul link capacity. Many uplink/downlink beamforming techniques were developed for cell-free networks. In \cite{Cell_Less_2}, the authors proposed the use of  conjugate beamforming on the downlink and matched filtering on the uplink to multiplex/de-multiplex signals to/from different users. 
It was also shown that using these beamforming techniques, the inter-user interference (IUI) goes to zero as the number of APs goes to infinity. 
Different uplink/downlink optimization techniques for cell-free decoding/precoding vectors that are based on maximizing the weighted sum-rate (WSR) or users' minimum rate were proposed in the literature. 
For instance, the authors in \cite{Cell_Free_Sum_Rate_1} proposed a low-complexity algorithm that solve the problem of maximizing the WSR of downlink cell-free system (mixed non-convex and combinatorial optimization problem) while maximizing the minimum per-user rate was modeled as a quasi-convex problem and solved by different iterative algorithms (such as bisection search and gradient descent) \cite{Cell_Less_1,Cell_Less_2,Cell_Less_3}. 

%\footnote{We may refer to this as successive interference cancellation (SIC) system rather than NOMA system since in cell-free networks, users already simultaneously transmitting and receiving throughout the same frequency band.}

To further enhance the performance of cell-free networks, non-orthogonal multiple access (NOMA) method can be used to suppress some of the in-band interfering signals
\cite{Cell_Free_NOMA_3,Cell_Free_NOMA_1,Cell_Free_NOMA_5}. 
In \cite{Cell_Free_NOMA_3}, downlink users were assumed to be divided into a number of NOMA clusters with successive interference cancellation (SIC) operation applied at every cluster members and the achievable per user transmission rate was calculated. The authors in \cite{Cell_Free_NOMA_1} proposed an adaptive NOMA/OMA selection scheme such that the downlink per-user transmission rate is maximized. It was also found that the performance of NOMA scheme outperforms that of cell-free OMA when the number of users is relatively high. Additionally, in \cite{Cell_Free_NOMA_5}, a spectral efficiency maximization algorithm for uplink NOMA-enabled cell-free network was proposed where the authors showed that a better performance can be achieved by controlling the per-user transmission power. 

%********* Talk about channel estimation-based deep learning for cell-free
%====================================================================================
\subsection{Motivation and Contributions}
%===================================================================================
%As can be inferred from the literature review presented in previous section, 

In the existing literature, performance analysis of cell-free networks has focused only on the derivation of an approximate mathematical limits of the per-user transmission rate.
This is due to the fact that other important  performance metrics such as the probability of outage [or alternatively coverage probability] are significantly difficult to be derived in a closed-form. Furthermore, for the cell-free networks, beamforming techniques (uplink and downlink) were basically developed and tested for relatively small number of users and APs. The reason is that computing the beamforming vectors for massive numbers of users and APs using the conventional algorithms (such as bisection search and gradient descent) gives rise to real-time implementation issues such as convergence time and computational complexity. In  \cite{GCoMP}, we proposed a downlink variable-order clustering scheme for APs that divides the APs into subgroups where every group is dedicated to serve a subset of the active users. This decreases the lengths of the beamforming vectors for every subgroup of APs by a factor that is linearly proportional to the clustering order. The performance degradation caused by clustering of the APs was compensated by adopting NOMA detection  in every cluster. Additionally,  clustering of APs was based on the channel gains of different users and APs.

This paper extends the work in \cite{GCoMP} for massive {\em uplink} multiple access in cell-free networks. For a static cell-free network, we first derive a closed-form expression of the per-user probability of outage. To reduce the complexity of joint processing of signals from all users in a static cell-free network, we propose a dynamic clustering scheme for APs. For real-time implementation of both dynamic AP clustering and uplink beamforming, we develop a deep reinforcement learning (DRL) scheme, namely, the hybrid Deep Deterministic Policy Gradient (DDPG)-deep double Q-network (DDQN)scheme. The major contributions on this paper can be summarized as follows:

\begin{itemize}
    \item For uplink static cell-free networks, we derive an accurate closed-form expression for the per-user probability of outage by exploiting the \textit{Welch-Satterthwaite} approximation \cite{Satterthwaite1946}.
    \item To  significantly decrease the signal processing complexity at the CPU for the static cell-free network, we propose a clustering scheme that dynamically partitions the APs into subsets with each subset acting as a virtual AP within a DAS system. We also propose an SIC-based signal detection scheme for non-orthogonal multiple access in a dynamic cell-free network and a modified DAS combining scheme that considers inter-user interference (IUI).
    \item We formulate a general problem to jointly optimize the clustering of APs and the beamforming vectors such that the uplink users' sum-rate is maximized or the minimum user rate is maximized.
    \item To solve the general optimization problem, we propose and design a novel hybrid DRL scheme based on DDPG-DDQN model.
    \item We study and compare different performance metrics of conventional static and those of dynamic cell-free networks under different number of users and APs.  
\end{itemize}
The rest of this paper is organized as follows.  The static cell-free network model and the analysis of outage performance for this network model are presented in Section II. In Section III, we present the dynamic cell-free network architecture and the corresponding outage performance analysis. For the dynamic cell-free network model, in Section IV, we present a successive interference cancellation (SIC)-aided signal detection scheme and a diversity combining scheme. Also, the joint optimization problem of AP clustering and  beamforming design for the dynamic cell-free model is formulated in this section. In Section V, we propose a novel hybrid DRL method that jointly performs AP clustering and optimization of the beamforming vectors. Numerical and simulation results are presented in Section VI
before the paper is concluded in Section VII. 

%===================================================================================
\textbf{Notations:} For a random variable (rv) $X$,  $F_X(x)$ and $f_X(x)$  represent cumulative distribution function (CDF) and probability density function (PDF), respectively.  $\mathbb P ( \cdot )$ and $\mathbb E[\cdot ]$ denote probability and  expectation. For a given matrix $\bm{A}\in \mathbb{C}^{M\times N}$, $A^H$ represents the Hermitian transpose of $\bm{A}$. A PDF expression of Nakagami-$\mathcal{M}$ rv is given by $f_X(x)=\frac{2\mathcal{M}^{\mathcal{M}}}{\Gamma(\mathcal{M}) {\Omega}^{\mathcal{M}}}x^{2\mathcal{M}-1}e^{\frac{\mathcal{M}}{\Omega}x^2}$ while a Gamma  rv is  denoted by $X \thicksim  \mathcal G (\alpha, \beta)$, with PDF as $f_{X}(x) = \frac{\beta^\alpha }{\Gamma(\alpha)} x^{\alpha -1} e^{- \beta x },  \quad x >0$, where $ \beta >0$, $ \alpha \geq  1 $, and $ \Gamma(z)$ is the Euler's Gamma function. The base of $\log(x) $ is  $2$. 
\section{Static Cell-Free Network Architecture}
%=============================================================

\subsection{Network Model}
We consider an uplink  network with $M$ single-antenna APs and $K$ single-antenna UEs that are located at fixed locations within a certain coverage area (Fig.~\ref{System_Model}).
\begin{figure}[htb]
		\centering
		\includegraphics[height=6cm, width=6cm]{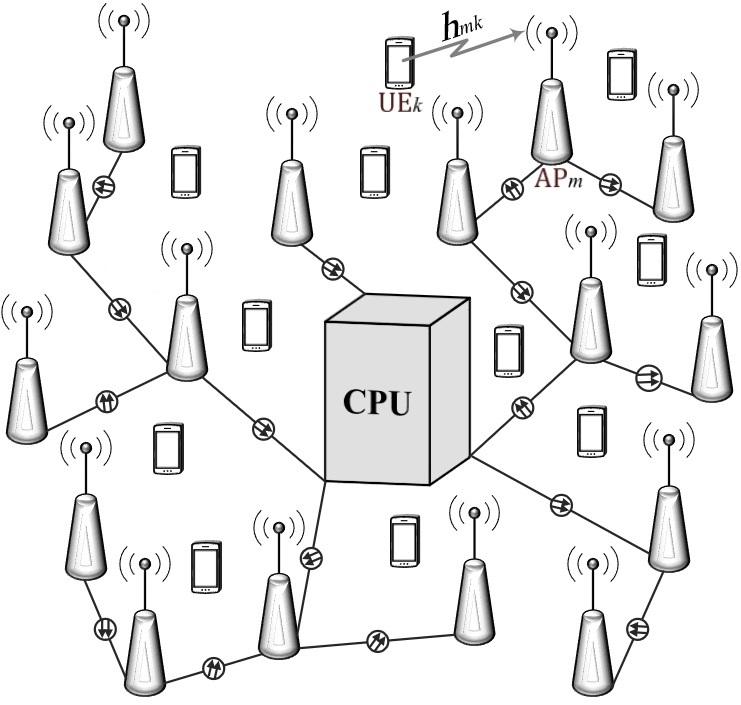}
		%\caption{Example scenario:(a) Static, (b) Dynamic. }\label{System_Model}
		\caption{Static cell-free network model.}\label{System_Model}
	\end{figure}
All APs are connected to each other through backhaul links to form a static cell-free network architecture \cite{Cell_Less_1}. 
This network setup enables the distributed APs to collaborate simultaneously to serve all users within the network coverage area.
The network is equipped with a CPU pool that performs the last-stage processing tasks to detect/decode signals from every user.	

The channel state information (CSI), which will be required at the CPU for signal detection for individual UEs, is assumed to be achieved through training the pilot sequences that are not completely orthogonal.
%This accurately evaluate the effect of CSI estimation errors under simultaneous massive processing of signals.
We assume that the channel gain between the $k\text{-th}$ user and the $m\text{-th}$ AP follows the following probabilistic model:
\begin{equation}\label{Channel_Model}
    g_{mk}={L}_{mk}^{-\kappa}\,h_{mk},
\end{equation}
where $L_{mk}=||d_{mk}||$ is the Euclidean distance between the $k\text{-th}$ user and the $m\text{-th}$ AP, $\kappa$ is the path-loss exponent ($\kappa\geq 2$) and $h_{mk}$ is the small-scale channel fading gain between the $k\text{-th}$ UE and the $m\text{-th}$ AP. 
We assume that $h_{mk}$ follows Nakagami-$\mathcal{M}_{mk}$ distribution with spreading and shape parameters $\mathcal{M}_{mk}$ and $\Omega_{mk}$, respectively, i.e. $|h_{mk}|^2\thicksim \mathcal{G}(\alpha_{mk},\beta_{mk})$, where $\alpha_{mk}=\mathcal{M}_{mk}$ and $\beta_{mk}=\frac{\mathcal{M}_{mk}}{\Omega_{mk}}$. 
%This assumption complies with the nature of next generations of wireless networks at which a large numbers of users and APs are densely distributed in a small scattered geographical areas which makes the use of significantly high range of frequencies (such as mmWaves and 5G new band [NB]) with the existence of strong signal component (not necessarily line-of-sight [LOS]). 
For simplicity of analysis, we assume that $L_{mk}^{-\kappa} (\forall~m$ and $k$) is known. We have $|g_{mk}|^2\thicksim \mathcal{G}(\alpha_{mk},\beta_{mk}/L_{mk}^{-2\kappa})$.
We assume that $g_{mk}(\forall m=1, \dots, M$ and $k=1, \dots, K$) belong to a set of independent but not identically distributed (i.n.d) rvs. 
%%*************************************************************************
%=======================================================================
\subsection{Uplink Network Training (CSI Acquisition)}
Under the assumption that large-scale fading gain ($L_{ {m}k}^{-\kappa}$) for all users is known, the aim of network training is to estimate the small-scale fading component of the overall channel gain ($h_{ {m}k}$). 
This can be achieved by assigning different pilot sequences for every user denoted by ${\bm{\varphi}}_{ {m} k}=\left[\varphi_{mk, 1} \dots \varphi_{mk, \tau_p} \right]^H$ such that $||{\bm{\varphi}}_{ {m} k}||^2=1$, where $\tau_p\leq \tau_c$ is the length of the length of pilot training sequence (in samples), which is less than or equal the channel coherence time ($\tau_c$). Accordingly, the received pilot vector at $m$-th AP is given by 
\begin{equation}
\label{Y_Pilot_Vector}
\bm{y}_{\textit{p}, {m} }=\sum_{k=1}^{K}\sqrt{\tau_p \rho_k}g_{ {m} k}\bm{\varphi}_{mk}+\bm{\eta}_{ {m} },
\end{equation}
where $\rho_k$ is the normalized transmitted power for each symbol of the $k$-th user's pilot vector and $\bm{\eta}_m\in \mathbb{C}^{\tau_p\times 1}$ is the additive white Gaussian noise (AWGN) vector related to pilot symbols with independent and identically distributed (i.i.d) rvs, i.e. ${\eta}_{i,m}\thicksim \mathcal{CG}\left(\mu_{i,m}, \sigma_{i,m}^2/2\right), \forall i=1, \dots, \tau_p$.
The goal is to find the best estimate of $g_{ {m} k}$ (denoted by $ \hat{g}_{ {m} k}$) given the vector of observations ($\bm{y}_{\textit{p}, {m} k}$). 
This can be optimally achieved by using the maximum {\em a posteriori} decision rule (MAP).
The Bayesian estimator of $g_{ {m} k}$ is found to be identical to that of the minimum mean square method (MMSE) \cite{Statistical_Signal_Processing, MMSE_1}.
Accordingly, the best estimate of $g_{ {m} k}$ can be expressed as \cite{Cell_Less_2}
\begin{dmath}
    \label{g_Estimate_1}
     \hat{g}_{ {m} k}=\frac{\sqrt{\tau_p \rho_k}L_{ {m} k}^{-\kappa}}{\rho_c\sum_{l=1, l\neq k}^{K}\rho_kL^{-\kappa}_{ {m} l}|\bm{\varphi}_{mk}^H\bm{\varphi}_{ml}|^2+1}\bm{\varphi}^H_{mk}\bm{y}_{\text{p}, {m} }
     \\
     = \mathcal{E}_{mk}\sqrt{\tau_p\rho_k}g_{mk}+\sum_{l=1, l\neq k}^K \mathcal{E}_{mk}\sqrt{\tau_p\rho_k}|\bm{\varphi}^H_{mk} \bm{\varphi}_{ml}|g_{ml}+\mathcal{E}_{mk}|\bm{\varphi}^H_{mk} \bm{\eta}_m|,
\end{dmath}
where $\mathcal{E}_{mk}=\frac{\sqrt{\tau_p \rho_k}L_{ {m} k}^{-\kappa}}{\rho_c\sum_{l=1, l\neq k}^{K}\rho_kL^{-\kappa}_{ {m} l}|\bm{\varphi}_{mk}^H\bm{\varphi}_{ml}|^2+1}$. 
Note that, if all users assigned a set of mutually orthogonal pilot sequences (i.e. $|\bm{\varphi}^H_{mk} \bm{\varphi}_{ml}|=0, k\neq l$), the estimated small-scale channel fading in (\ref{g_Estimate_1}) reduces to a scaled version of the exact fading gain plus a relatively small AWGN noise portion.
However, for certain applications (e.g. mMTC applications) and due to length limitations of $\tau_p$, non-orthogonal pilot signals has to be used among some active users.

%%**************************************************************************
%===================================================================================================
\subsection{Outage Performance}
\subsubsection{Distribution of Users' SINR}
In the static cell-free uplink network with $M$ APs and $K$ users, the distributed APs receive signals from all users within their coverage area and then forward the baseband signal components through the fronthaul links to the CPU pool. The CPU performs the detection of each user's signal separately. This process can be implemented by optimizing a user's transmission power and/or the beamforming vectors at every user's detector.
At a certain CPU detector module, the received signal used in the detection of the $k\text{-th}$ user's component is given by
\begin{dmath}
 y_k=\sum_{m=1}^Mw_{mk}\left[\sum_{l=1}^K\hat{g}_{ml}\sqrt{p_l}x_l+\tilde{\eta}_m\right]
    =
    \underbrace{
    \sqrt{\tau_p\rho_kp_k}x_k\sum_{m=1}^Mw_{mk}\mathcal{E}_{mk}g_{mk}}
    _{\text{Desired Signal}}
    +
    \sum_{m=1}^M
    w_{mk}
     \left[
     \underbrace{    \sum_{l=1, l\neq k}^K\sqrt{\tau_p\rho_lp_l}x_l\mathcal{E}_{ml}g_{ml}
    }
    _{\text{Inter-User Interference}}
    \\
    +
    \underbrace{
    \sum_{l=1, l\neq k}^K\sqrt{\tau_p\rho_kp_k}x_k\mathcal{E}_{mk}|\bm{\varphi}_{mk}^H\bm{\varphi}_{ml}|g_{ml}
    +
    \sum_{l=1, l\neq k}^K \sum_{\tilde{l}=1, \tilde{l}\neq l}^K\sqrt{\tau_p\rho_lp_l}x_l\mathcal{E}_{ml}|\bm{\varphi}_{ml}^H\bm{\varphi}_{m\tilde{l}}|g_{m\tilde{l}}}
    _{\text{Non-Orthogonal Pilots' Related Estimation error}}
        \right]
    \\
    +
    \sum_{m=1}^M
    w_{mk}
    \left[
          \underbrace{ \sqrt{p_k}x_k|\bm{\varphi}_{mk}^H\bm{\eta}_m|
    +
    \sum_{l=1, l\neq k}^K\sqrt{p_l}x_l|\bm{\varphi}_{ml}^H\bm{\eta}_m|
    }
    _{\text{AWGN's Related Estimation Error}}
    +
              \underbrace{  \tilde{\eta}_m}
    _{\text{AWGN Component}}
    \right]
    ,\label{Received_Signal_1} 
\end{dmath}
where $w_{mk}$ is the $m$-th element of the $k$-th user's beamforming vector such that $0\leq w_{mk}\leq 1$, $p_k$ is the uplink transmission power of the $k$-th user such that $0\leq p_k\leq P_k$, where $P_k$ is the maximum allowable transmission power for the $k\text{-th}$ user, $x_k$ is the transmitted symbol from the $k\text{-th}$ user such that $\mathbb{E}[|x_k|^2]=1$ and ${\tilde{\eta}}_m$ is the AWGN at the input of the $m$-th AP such that $\tilde{\eta}_m\thicksim \mathcal{CG}\left(\tilde{\mu}_m, \tilde{\sigma}_m^2/2\right)$. 
We assume that $\tilde{{\eta}}_m, \forall m=1, \dots, M$ are from a set of i.i.d rvs.

Accordingly, the SINR of the $k\text{-th}$ user related to $y_k$ at (\ref{Received_Signal_1}) can be expressed as
\begin{dmath}
    \gamma_k=
    \frac
    {
    \sum_{m=1}^M|\tilde{g}_{mk}|^2
    }
    {
    \sum_{m=1}^M
    \left[
    \sum_{l=1, l\neq k}^K|\tilde{g}_{ml}|^2
    +
    \sum_{\dot{l}=1, \dot{l}\neq k}^K
    \sum_{\ddot{l}=1, \ddot{l}\neq \dot{l}}^K
    |\tilde{g}_{m{\ddot{l}}}|^2
    +
    \sum_{\Breve{l}=1, \Breve{l}\neq k}^K|\tilde{g}_{m{\Breve{l}}}|^2
    \right]+1
    }
%    =
%    \frac{
%    \sum_{m=1}^M|\tilde{g}_{mk}|^2
 %   }
%    {
 %   \sum_{m=1}^M\sum_{l=1, l\neq k}^K|\tilde{g}_{mk}|^2+1
%    }
,\label{SINR_1}
\end{dmath}
where $|\tilde{g}_{mk}|^2\thicksim \mathcal{G}\left(\tilde{\alpha}_{mk},\tilde{\beta}_{mk} \right)$, $|\tilde{g}_{ml}|^2\thicksim \mathcal{G}\left(\tilde{\alpha}_{ml},\tilde{\beta}_{ml} \right)$, $|\tilde{g}_{m{\dot{l}}}|^2\thicksim \mathcal{G}\left(\tilde{\alpha}_{m\dot{l}},\tilde{\beta}_{m\dot{l}} \right)$,  $|\tilde{g}_{m\ddot{l}}|^2\thicksim \mathcal{G}\left(\tilde{\alpha}_{m\ddot{l}},\tilde{\beta}_{m\ddot{l}} \right)$ and 
$|\tilde{g}_{m\Breve{l}}|^2\thicksim \mathcal{G}\left(\tilde{\alpha}_{m\Breve{l}},\tilde{\beta}_{m\Breve{l}} \right)$
 with
 $\tilde{\alpha}_{mk}=\mathcal{M}_{mk}$, 
 $\tilde{\beta}_{mk}=\frac{\mathcal{M}_{mk}\dot{\sigma}_{mk}L_{mk}^{2\kappa}}{\Omega_{mk}    w_{mk}^2\tau_p\rho_kp_k\mathcal{E}_{mk}^2}$, 
 $\tilde{\alpha}_{ml}=\mathcal{M}_{ml}$,  
 $\tilde{\beta}_{ml}=\frac{\mathcal{M}_{ml}\dot{\sigma}_{mk}L_{ml}^{2\kappa} }{\Omega_{ml}  w_{mk}^2\tau_p\rho_kp_k\mathcal{E}_{mk}^2|\bm{\varphi}_{mk}^H\bm{\varphi}_{ml}|^2}$, 
 $\tilde{\alpha}_{m\ddot{l}}=\mathcal{M}_{m\ddot{l}}$,  
 $\tilde{\beta}_{m\ddot{l}}=\frac{\mathcal{M}_{m\ddot{l}}\dot{\sigma}_{mk}L_{m\ddot{l}}^{2\kappa}}{\Omega_{m\ddot{l}}   w_{mk}^2\tau_p\rho_{\dot{l}}p_{\dot{l}}\mathcal{E}_{m{\dot{l}}}^2|\bm{\varphi}_{m\dot{l}}^H\bm{\varphi}_{m\ddot{l}}|^2}$,
 $\tilde{\alpha}_{m\Breve{l}}=\mathcal{M}_{m\Breve{l}}$,  
 $\tilde{\beta}_{m\Breve{l}}=\frac{\mathcal{M}_{m\Breve{l}}\dot{\sigma}_{mk}L_{m\Breve{l}}^{2\kappa}}{\Omega_{m\Breve{l}}   w_{mk}^2\tau_p\rho_{\Breve{l}}p_{\dot{l}}\mathcal{E}_{m{\dot{l}}}^2}$.\\ 
 $\dot{\sigma}_{mk}=\sum_{m=1}^M\left[\sum_{t=1}^{\tau_p}\frac{w_{mk}^2\sigma_m^2}{2}\left(p_k\varphi_{mk,t}+\sum_{l=1, l\neq k}^Kp_l\varphi_{ml,t}\right)+ w_{mk}^2\tilde{\sigma}^2_m/2\right]$.

Note that the cross-product terms between channel gains are averaged to zero due to the i.i.d assumption. 
Furthermore, the AWGN term in the denominator (denoted by $\dot{\sigma}_m$) is replaced by the PSD (variance) of the sum of $M$ i.i.d rvs. 
Our goal is to first find a closed-form expression for the PDF and CDF of $\gamma_k$ and then utilize the derived expressions in deriving some fundamental performance limits for the static cell-free network. \textbf{Lemma~\ref{l_pdfs}} gives an accurate approximation for the PDF of $\gamma_k$ in (\ref{SINR_1}).
	\begin{lemma} \label{l_pdfs}
    Let 
    $ |\tilde{g}_{mk}|^2 \thicksim \gd {\tilde{\alpha}_{mk}} {\tilde{\beta}_{mk}} $
    , 
    $ |\tilde{g}_{ml}|^2 \thicksim \gd {\tilde{\alpha}_{ml}} {\tilde{\beta}_{ml}} $
    , 
    $ |\tilde{g}_{m\ddot{l}}|^2 \thicksim \gd {\tilde{\alpha}_{m\ddot{l}}} {\tilde{\beta}_{m\ddot{l}}} $ 
    and 
    $ |\tilde{g}_{m\Breve{l}}|^2 \thicksim \gd {\tilde{\alpha}_{m\Breve{l}}} {\tilde{\beta}_{m\Breve{l}}} $
    where 
    $|\tilde{g}_{mk}|^2$
    ,
    $|\tilde{g}_{ml}|^2$
    ,
    $|\tilde{g}_{m\ddot{l}}|^2$
    and
    $|\tilde{g}_{m\Breve{l}}|^2$ are independent rvs with $m=1, \dots, M$ and $k=1, \dots, K$.  The PDF of  $ \gamma_k$ in (\ref{SINR_1}) with relatively large $M$ and $K$ is 
    \begin{dmath}	f_{\gamma_k}\left(\gamma\right)=
    \frac{\Dot{\beta}_{mk}^{\Dot{\alpha}_{mk}}\Dot{\beta}_{mk'}^{\Dot{\alpha}_{mk'}}e^{-\frac{\Dot{\beta}_{mk}\gamma-\Dot{\beta}_{mk'}}{2}}}
    {\Gamma\left(\Dot{\alpha}_{mk}\right)\left(\Dot{\beta}_{mk}\gamma+\Dot{\beta}_{mk'} \right)^{\frac{\Dot{\alpha}_{mk}+\Dot{\alpha}_{ml}+1}{2}}}
    W_{\frac{\Dot{\alpha}_{mk}-\Dot{\alpha}_{mk'}+1}{2},\frac{-\Dot{\alpha}_{mk}-\Dot{\alpha}_{mk'}}{2}}\left(\Dot{\beta}_{mk}\gamma+\Dot{\beta}_{mk'} \right) 
    ,\label{PDF_Centralized}		\end{dmath}
		where  
        $\Dot{\alpha}_{mk}$, $\Dot{\alpha}_{mk'}$, $\Dot{\beta}_{mk}$ and $\Dot{\beta}_{mk'}$ are defined in \textbf{Appendix A} and  $W_{\lambda,\mu}(x)$ is the  Whittaker function with parameters $\lambda$ and $\mu$ \cite[Eq. 9.222.1]{gradshteyn2000}.
	\end{lemma}
    \begin{proof}
See \textbf{Appendix A}.
\end{proof}

%================================================================
\subsubsection{SINR Outage}
A receiver is in an outage if the SINR of the received signal falls bellow a certain predefined threshold value (denoted by $\gamma_{\text{th}}$). 
\textbf{Theorem \ref{Theorem_Outage}} below gives an accurate approximation for the average probability of outage of the $k
$-th user.
	\begin{theorem}\label{Theorem_Outage} \label{l_cdfs}
    If 
    $ |\tilde{g}_{mk}|^2 \thicksim \gd {\tilde{\alpha}_{mk}} {\tilde{\beta}_{mk}} $
    , 
    $ |\tilde{g}_{ml}|^2 \thicksim \gd {\tilde{\alpha}_{ml}} {\tilde{\beta}_{ml}} $
    , 
    $ |\tilde{g}_{m\ddot{l}}|^2 \thicksim \gd {\tilde{\alpha}_{m\ddot{l}}} {\tilde{\beta}_{m\ddot{l}}} $ 
    and 
    $ |\tilde{g}_{m\Breve{l}}|^2 \thicksim \gd {\tilde{\alpha}_{m\Breve{l}}} {\tilde{\beta}_{m\Breve{l}}} $
    where 
    $|\tilde{g}_{mk}|^2$
    ,
    $|\tilde{g}_{ml}|^2$
    ,
    $|\tilde{g}_{m\ddot{l}}|^2$
    and
    $|\tilde{g}_{m\Breve{l}}|^2$ are independent rvs with $m=1, \dots, M$ and $k=1, \dots, K$.  The probability of outage of the $k$-th user in a static cell-free network is 
    \begin{dmath}	P_{\text{out}}^{(k)}=1-\frac{\left(\frac{\Dot{\beta}_{mk'}}{\gamma_{\text{th}}\Dot{\beta}_{mk}} \right)^{\Dot{\alpha}_{mk'}}\Gamma\left(\Dot{\alpha}_{mk}+\Dot{\alpha}_{mk'} \right)}{\Dot{\alpha}_{mk'}\Gamma\left(\Dot{\alpha}_{mk} \right)\Gamma\left(\Dot{\alpha}_{mk'} \right)}
    {}_{2}F_{1}\left(\Dot{\alpha}_{mk'},\Dot{\alpha}_{mk}+\Dot{\alpha}_{mk'};1+\Dot{\alpha}_{mk'};-\frac{\Dot{\beta}_{mk'}}{\Dot{\beta}_{mk}} \frac{1}{\gamma_{\text{th}}}\right)
    ,\label{e_outage_1}		\end{dmath}
		where  
        $\gamma_{\text{th}}$ is the SINR threshold value and  ${}_2F_{1}(.)$ is the  Gauss hypergeometric function function \cite[Eq. 9.11.1]{gradshteyn2000}.
	\end{theorem}
    \begin{proof}
See \textbf{Appendix B}.   
\end{proof}

One major drawback of the static cell-free network architecture is the requirement of joint processing of signals related to all active users which will require the CSI for the links from all users to all APs. This becomes very challenging for massive numbers of active users and APs in the network. In the following section, we propose a dynamic cell-free network model for uplink access that reduces the complexity of joint processing, compared to that of static cell-free network model, in order to serve a massive number of users. 

%=============================================================
\section{Dynamic Cell-Free Architecture}
%=============================================================

%In \cite{GCoMP}, we designed a novel {\em downlink} cell-free network that partitions the users and APs based on the CSI into groups of clusters with each cluster containing a set of users and serving APs that act as a mini cell-free network with centralized processing. 
%The size and content of each cluster is specified by the current CSI of the overall network. 
%The proposed system showed a significant complexity reduction at the expense of slight degradation in system performance.

%==========================================================================================================================

\subsection{Network Model and Assumptions}\label{s_sys}
In the proposed dynamic cell-free network model (Fig.~\ref{System_Model_2}), the  APs in the network are partitioned among a set of $\tilde{M}$ subgroups (known as clusters) in which every cluster consists of $\mathcal{N}_{\tilde{m}}$ APs such that $1\leq \mathcal{N}_{\tilde{m}}\leq M-(\tilde{M}-1)$ \footnote{Note that when $\tilde{M}=M$, every cluster will contain a single AP and this falls back to the conventional static cell-free network model.}.
Additionally, the $\tilde{m}$-th cluster, where $\tilde{m}=1, \dots, \tilde{M}$,  acts as an single virtual AP with a $\mathcal{N}_{\tilde{m}}$-antenna DAS that combine the overall received signals using a predefined combining technique (such as maximal ratio combining [MRC], equal gain combining [EGC], selection combining [SC], etc.)\cite{simonmk05}.
The clustering of APs is assumed to be performed in each time slot (or transmission interval) based on the current CSI of the network.
%This can be achieved by optimizing the APs' clustering set, the beamforming vectors and the DAP combining mechanism.   
Similar to the static cell-free network model, we assume that channel gain between the $k$-th user and the ${n}_{\tilde{m}}$-th antenna at the $\tilde{m}$-th cluster (the virtual AP) follows the following distribution:
\begin{equation}
    g_{\tilde{m}n_{\tilde{m}}k}={L}_{\tilde{m}n_{\tilde{m}}k}^{-\kappa}h_{\tilde{m}n_{\tilde{m}}k},
\end{equation}
where $L_{\tilde{m}n_{\tilde{m}}k}=||d_{\tilde{m}n_{\tilde{m}}k}||$ is the Euclidean distance between the $k\text{-th}$ user and the $n_{\tilde{m}}$-th antenna at the $\tilde{m}$-th cluster, $\kappa$ is the path-loss exponent with a value depends on the propagation environment with $\kappa\geq 2$ and $h_{\tilde{m}n_{\tilde{m}}k}$ is the small-scale channel fading between the $k\text{-th}$ user and the $n_{\tilde{m}}$-th antenna at the $\tilde{m}$-th cluster. 
Similarly,  $h_{\tilde{m}n_{\tilde{m}}k}$ is assumed to follow a Nakagami-$\mathcal{M}_{\tilde{m}n_{\tilde{m}}k}$ distribution, i.e. $|h_{\tilde{m}n_{\tilde{m}}k}|^2\thicksim \mathcal{G}(\alpha_{\tilde{m}n_{\tilde{m}}k},\beta_{\tilde{m}n_{\tilde{m}}k})$, where $\alpha_{\tilde{m}n_{\tilde{m}}k}=\mathcal{M}_{\tilde{m}n_{\tilde{m}}k}$ represents the shape parameter of Nakagami-$\mathcal{M}_{\tilde{m}n_{\tilde{m}}k}$ fading and $\beta_{\tilde{m}n_{\tilde{m}}k}=\frac{\mathcal{M}_{mnk}}{\Omega_{\tilde{m}n_{\tilde{m}}k}}$ is defined as the shape parameter normalized by spreading parameter of Nakagami-$\mathcal{M}_{\tilde{m}n_{\tilde{m}}k}$ fading. 
Again, for simplicity of analysis, we assume that $L_{\tilde{m}n_{\tilde{m}}k}^{-\kappa}, \forall~\tilde{m},n_{\tilde{m}},k$ is known. Accordingly, we have $|g_{\tilde{m}n_{\tilde{m}}k}|^2\thicksim \mathcal{G}(\alpha_{\tilde{m}n_{\tilde{m}}k},\beta_{\tilde{m}n_{\tilde{m}}k}')$, where $\beta_{\tilde{m}n_{\tilde{m}}k}'=\beta_{\tilde{m}n_{\tilde{m}}k}/L_{\tilde{m}n_{\tilde{m}}k}^{-2\kappa}$. 
Let us denote $\bm{\mathcal{C}}=\{\mathcal{C}_1 \dots \mathcal{C}_N \}$ as the set of all possible clustering configurations of APs such that every cluster contains at least one  AP. 
As an example, with $M=8$ and $\tilde{M}=3$, one possible set is
\[
\mathcal{C}_n=\{ \underbrace{ \{ AP_3, AP_6, AP_8 \}}_{\mathcal{N}_{\tilde{m}=1}=3}, \dots, \underbrace{\{AP_2\}}_{\mathcal{N}_{\tilde{m}=2}=1},\underbrace{\{ \overbrace{AP_1}^{n_{\tilde{m}=3}=1}, AP_4, AP_5, AP_7 \}}_{\mathcal{N}_{\tilde{m}=\tilde{M}=3}=4}  \}.
\]
%=========================================================================

\begin{figure}[htb]
		\centering
		\includegraphics[height=6cm, width=6.5cm]{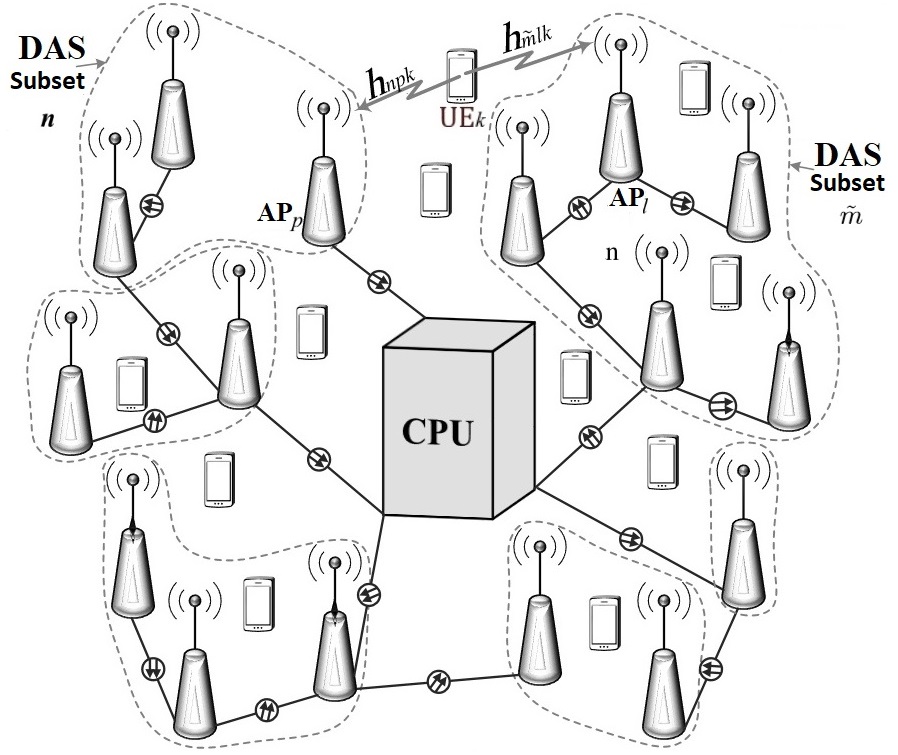}
		%\caption{Example scenario:(a) Static, (b) Dynamic. }\label{System_Model}
		\caption{Dynamic cell-free network model.}\label{System_Model_2}
	\end{figure}
Note that in this network setup, the CSI and the instantaneous clustering information will be only required at the CPU for decoding the signals from the UEs. 

%========================================================
\subsection{Outage Performance}
%========================================================================
For the proposed dynamic cell-free model, the received signal used in the detection of the $k$-th user's component is given by  
\begin{dmath}
\scriptsize
y_k=
\sum_{\tilde{m}=1}^{\tilde{M}}w_{\tilde{m}k}
\sum_{n_{\tilde{m}}=1}^{\mathcal{N}_{\tilde{m}}}G_{\tilde{m}n_{\tilde{m}}k}
\left[\sum_{l=1}^K\hat{g}_{\tilde{m}n_{\tilde{m}}l}\sqrt{p_l}x_l+{\tilde{\eta}}_{\tilde{m}n_{\tilde{m}}}\right]
  =
    \underbrace{
    \sqrt{\tau_p\rho_kp_k}x_k\sum_{\tilde{m}=1}^{\tilde{M}}w_{\tilde{m}k}\mathcal{E}_{\tilde{m}n_{\tilde{m}}k}
    \sum_{n_{n_{\tilde{m}}}}^{\mathcal{N}_{n_{\tilde{m}}}}G_{{\tilde{m}}n_{\tilde{m}}k}
    g_{\tilde{m}n_{\tilde{m}}k}}
    _{\text{Desired Signal}}+
    \sum_{\tilde{m}=1}^{\tilde{M}}
    w_{\tilde{m}k}
    \sum_{n_{{\tilde{m}}}=1}^{\mathcal{N}_{\tilde{m}}}G_{{\tilde{m}}n_{\tilde{m}}k}
     \left[
     \underbrace{    \sum_{l=1, l\neq k}^K\sqrt{\tau_p\rho_lp_l}x_l\mathcal{E}_{\tilde{m}n_{\tilde{m}}l}g_{\tilde{m}n_{\tilde{m}}l}
    }
    _{\text{Inter-User Interference}}
    \\
    +
    \underbrace{
    \sum_{l=1, l\neq k}^K\sqrt{\tau_p\rho_kp_k}x_k\mathcal{E}_{\tilde{m}n_{\tilde{m}}k}|\bm{\varphi}_{\tilde{m}n_{\tilde{m}}k}^H\bm{\varphi}_{\tilde{m}n_{\tilde{m}}l}|g_{\tilde{m}n_{\tilde{m}}l}
    +
    \sum_{l=1, l\neq k}^K \sum_{\tilde{l}=1, \tilde{l}\neq l}^K\sqrt{\tau_p\rho_lp_l}x_l\mathcal{E}_{\tilde{m}n_{\tilde{m}}l}|\bm{\varphi}_{\tilde{m}n_{\tilde{m}}l}^H\bm{\varphi}_{\tilde{m}n_{\tilde{m}}{\tilde{l}}}|g_{\tilde{m}n_{\tilde{m}}\tilde{l}}}
    _{\text{Non-Orthogonal Pilots' Related Estimation error}}
        \right]
    \\
    +
    \sum_{\tilde{m}=1}^{\tilde{M}}
    w_{\tilde{m}k}
    \sum_{n_{\tilde{m}}=1}^{\mathcal{N}_{\tilde{m}}}G_{{\tilde{m}}n_{\tilde{m}}k}
    \left[
          \underbrace{ \sqrt{p_k}x_k|\bm{\varphi}_{\tilde{m}n_{\tilde{m}}k}^H\bm{\eta}_{\tilde{m}n_{\tilde{m}}}|
    +
    \sum_{l=1, l\neq k}^K\sqrt{p_l}x_l|\bm{\varphi}_{\tilde{m}n_{\tilde{m}}l}^H\bm{\eta}_{\tilde{m}n_{\tilde{m}}}|
    }
    _{\text{AWGN's Related Estimation Error}}
    +
              \underbrace{  \tilde{\eta}_{\tilde{m}n_{\tilde{m}}}}
    _{\text{AWGN}}
    \right], \label{Received_Signal_2} 
\end{dmath}
where $w_{\tilde{m}k}$ is the $\tilde{m}$-th element of the $k$-th user's beamforming vector such that $0\leq w_{{\tilde{m}}k}\leq 1$, 
$G_{\tilde{m}n_{\tilde{m}}k}$ is a design gain parameter for the $\mathcal{N}_{\tilde{m}}$-antenna DAS at the $\tilde{m}$-th virtual AP,
$p_{k}$ is the uplink transmission power of the $k$-th user such that $0\leq p_k\leq P_k$, where $P_k$ is the power budget of the $k\text{-th}$ user, $x_k$ is the transmitted symbol from the $k\text{-th}$ user such that $\mathbb{E}[|x_k|^2]=1$, and $\tilde{\eta}_{\tilde{m}n_{\tilde{m}}}$ is the AWGN at the $m$-th AP with $\tilde{\eta}_{\tilde{m}n_{\tilde{m}}}\thicksim \mathcal{CG}\left(\tilde{\mu}_{\tilde{m}n_{\tilde{m}}}, \tilde{\sigma}_{\tilde{m}n_{\tilde{m}}}^2/2\right)$.
We assume that $\tilde{\eta}_{\tilde{m}n_{\tilde{m}}}, \forall \{\tilde{m}~ \&~ n_{\tilde{m}}\}$ are from a set of i.n.d rvs. 

For uplink network training, we follow a procedure similar to that in Section II.B. 
Therefore, we have ${\bm{\varphi}}_{\tilde{m}n_{\tilde{m}}k}=\left[\varphi_{{\tilde{m}n_{\tilde{m}}k}, 1} \dots \varphi_{{\tilde{m}n_{\tilde{m}}k}, \tau_p} \right]^H$  as the set of $\tau_p$-dimensional training symbols such that $||{\bm{\varphi}}_{\tilde{m}n_{\tilde{m}}k}||^2=1$ and $\mathcal{E}_{\tilde{m}n_{\tilde{m}}k}=\frac{\sqrt{\tau_p \rho_k}L_{\tilde{m}n_{\tilde{m}}k}^{-\kappa}}{\rho_c\sum_{l=1, l\neq k}^{K}\rho_kL^{-\kappa}_{\tilde{m}n_{\tilde{m}}l}|\bm{\varphi}_{\tilde{m}n_{\tilde{m}}k}^H\bm{\varphi}_{\tilde{m}n_{\tilde{m}}l}|^2+1}$ is the MMSE channel estimation constant. 
Note that for every cluster set ($\mathcal{C}_n, n=1, \dots, N$), different SINR values ($\gamma_k$) will result  for different users.

%In a similar procedure to the case of  fully centralized cell-free networks, 

The SINR $\gamma_k$ of the $k$-th user related to $y_k$ in (\ref{Received_Signal_2}) and under the $n$-th cluster can be expressed as   
\begin{dmath}
   \text{\hspace{-3mm}} \gamma_k^{\{\mathcal{C}_n\}}=
    \frac
    {
    \sum_{\tilde{m}=1}^{\tilde{M}}
    \sum_{n_{\bar{m}}=1}^{\mathcal{N}_{\bar{m}}}
    |\tilde{g}_{\tilde{m}n_{\tilde{m}}k}|^2
    }
    {
    \sum_{\tilde{m}=1}^{\tilde{M}}
    \sum_{n_{\tilde{m}}=1}^{\mathcal{N}_{\tilde{m}}}
    \left[
    \sum_{l=1, l\neq k}^K|\tilde{g}_{\tilde{m}n_{\tilde{m}}l}|^2
    +
    \sum_{\dot{l}=1, \dot{l}\neq k}^K
    \sum_{\ddot{l}=1, \ddot{l}\neq \dot{l}}^K
    |\tilde{g}_{{\tilde{m}n_{\tilde{m}}}{\ddot{l}}}|^2
    +
    \sum_{\Breve{l}=1, \Breve{l}\neq k}^K|\tilde{g}_{{\tilde{m}n_{\tilde{m}}}{\Breve{l}}}|^2
    \right]+1
    } 
,\label{SINR_2}
\end{dmath}
 where
    $ |\tilde{g}_{\tilde{m}n_{\tilde{m}}k}|^2 \thicksim \gd {\tilde{\alpha}_{\tilde{m}n_{\tilde{m}}k}} {\tilde{\beta}_{\tilde{m}n_{\tilde{m}}k}} $
    , 
    $ |\tilde{g}_{\tilde{m}n_{\tilde{m}}l}|^2 \thicksim \gd {\tilde{\alpha}_{\tilde{m}n_{\tilde{m}}l}} {\tilde{\beta}_{\tilde{m}n_{\tilde{m}}l}} $
    , 
    $ |\tilde{g}_{{\tilde{m}n_{\tilde{m}}}\ddot{l}}|^2 \thicksim \gd {\tilde{\alpha}_{{\tilde{m}n_{\tilde{m}}}\ddot{l}}} {\tilde{\beta}_{{\tilde{m}n_{\tilde{m}}}\ddot{l}}} $ 
    and\\ 
    $ |\tilde{g}_{{\tilde{m}n_{\tilde{m}}}\Breve{l}}|^2 \thicksim \gd {\tilde{\alpha}_{{\tilde{m}n_{\tilde{m}}}\Breve{l}}} {\tilde{\beta}_{{\tilde{m}n_{\tilde{m}}}\Breve{l}}} $ 
 with
 $\tilde{\alpha}_{\tilde{m}n_{\tilde{m}}k}=\mathcal{M}_{\tilde{m}n_{\tilde{m}}k}$, 
 
 $\tilde{\beta}_{\tilde{m}n_{\tilde{m}}k}=\frac{\mathcal{M}_{\tilde{m}n_{\tilde{m}}k}\dot{\sigma}_{\tilde{m}n_{\tilde{m}}k}L_{\tilde{m}n_{\tilde{m}}k}^{2\kappa}}{\Omega_{\tilde{m}n_{\tilde{m}}k}    w_{mk}^2G_{\tilde{m}n_{\tilde{m}}k}\tau_p\rho_kp_k\mathcal{E}_{\tilde{m}n_{\tilde{m}}k}^2}$, 
 $\tilde{\alpha}_{\tilde{m}n_{\tilde{m}}l}=\mathcal{M}_{\tilde{m}n_{\tilde{m}}l}$,  
 
 $\tilde{\beta}_{\tilde{m}n_{\tilde{m}}l}=\frac{\mathcal{M}_{\tilde{m}n_{\tilde{m}}l}\dot{\sigma}_{\tilde{m}n_{\tilde{m}}l}L_{\tilde{m}n_{\tilde{m}}l}^{2\kappa} }{\Omega_{\tilde{m}n_{\tilde{m}}l}  w_{mk}^2G_{\tilde{m}n_{\tilde{m}}k}\tau_p\rho_kp_k\mathcal{E}_{\tilde{m}n_{\tilde{m}}k}^2|\bm{\varphi}_{\tilde{m}n_{\tilde{m}}k}^H\bm{\varphi}_{\tilde{m}n_{\tilde{m}}l}|^2}$, 
 $\tilde{\alpha}_{{\tilde{m}n_{\tilde{m}}}\ddot{l}}=\mathcal{M}_{{\tilde{m}n_{\tilde{m}}}\ddot{l}}$, 
 
 $\tilde{\beta}_{{\tilde{m}n_{\tilde{m}}}\ddot{l}}=\frac{\mathcal{M}_{{\tilde{m}n_{\tilde{m}}}\ddot{l}}\dot{\sigma}_{{\tilde{m}n_{\tilde{m}}}k}L_{{{\tilde{m}n_{\tilde{m}}}}\ddot{l}}^{2\kappa}}{\Omega_{{{\tilde{m}n_{\tilde{m}}}}\ddot{l}}   w_{mk}^2G_{\tilde{m}n_{\tilde{m}}k}\tau_p\rho_{\dot{l}}p_{\dot{l}}\mathcal{E}_{{{\tilde{m}n_{\tilde{m}}}}{\dot{l}}}^2|\bm{\varphi}_{{{\tilde{m}n_{\tilde{m}}}}\dot{l}}^H\bm{\varphi}_{{{\tilde{m}n_{\tilde{m}}}}\ddot{l}}|^2}$,
 
 $\tilde{\alpha}_{{{\tilde{m}n_{\tilde{m}}}}\Breve{l}}=\mathcal{M}_{{{\tilde{m}n_{\tilde{m}}}}\Breve{l}}$,  
 
 $\tilde{\beta}_{{{\tilde{m}n_{\tilde{m}}}}\Breve{l}}=\frac{\mathcal{M}_{{{\tilde{m}n_{\tilde{m}}}}\Breve{l}}\dot{\sigma}_{{\tilde{m}n_{\tilde{m}}}k}L_{{{\tilde{m}n_{\tilde{m}}}}\Breve{l}}^{2\kappa}}{\Omega_{{{\tilde{m}n_{\tilde{m}}}}\Breve{l}}   w_{mk}^2\tau_p\rho_{\Breve{l}}p_{\dot{l}}\mathcal{E}_{{{\tilde{m}n_{\tilde{m}}}}{\dot{l}}}^2}$, and
 \\
 $\dot{\sigma}_{\tilde{m}n_{\tilde{m}}k}=\sum_{\tilde{m}=1}^{\tilde{M}}
 w_{\tilde{m}k}^2
 \sum_{n_{\tilde{m}}=1}^{\mathcal{N}_{\tilde{m}}}G_{\tilde{m}n_{\tilde{m}}k}^2
 \left[\sum_{t=1}^{\tau_p}\frac{\sigma_{\tilde{m}n_{\tilde{m}}}^2}{2}\left(p_k\varphi_{{\tilde{m}n_{\tilde{m}}k},t}+\sum_{l=1, l\neq k}^Kp_l\varphi_{{\tilde{m}n_{\tilde{m}}l},t}\right)+ \tilde{\sigma}^2_{\tilde{m}n_{\tilde{m}}}/2\right]$.\\
Note that $G_{\tilde{m}nk}$ is a one-time design parameter while $w_{\tilde{m}k}$ is an optimization parameter.
Generally, different performance metrics of the dynamic cell-free network will depend on the clustering algorithm that assigns different APs to different clusters.
\textbf{Corollary \ref{Dynamic_Performance}} below states the outage performance of a dynamic cell-free network under a random clustering scheme.
%at which APs are assigned randomly among a random number of clusters. 
\begin{corollary} \label{Dynamic_Performance}
     If 
    $ |\tilde{g}_{\tilde{m}n_{\tilde{m}}k}|^2 \thicksim \gd {\tilde{\alpha}_{\tilde{m}n_{\tilde{m}}k}} {\tilde{\beta}_{\tilde{m}n_{\tilde{m}}k}} $
    , ~
    $ |\tilde{g}_{\tilde{m}n_{\tilde{m}}l}|^2 \thicksim \gd {\tilde{\alpha}_{\tilde{m}n_{\tilde{m}}l}} {\tilde{\beta}_{\tilde{m}n_{\tilde{m}}l}} $
    , ~
    $ |\tilde{g}_{{\tilde{m}n_{\tilde{m}}}\ddot{l}}|^2 \thicksim \gd {\tilde{\alpha}_{{\tilde{m}n_{\tilde{m}}}\ddot{l}}} {\tilde{\beta}_{{\tilde{m}n_{\tilde{m}}}\ddot{l}}} $ 
    and 
    $ |\tilde{g}_{{\tilde{m}n_{\tilde{m}}}\Breve{l}}|^2 \thicksim \gd {\tilde{\alpha}_{{\tilde{m}n_{\tilde{m}}}\Breve{l}}} {\tilde{\beta}_{{\tilde{m}n_{\tilde{m}}}\Breve{l}}} $
    where 
    $|\tilde{g}_{\tilde{m}n_{\tilde{m}}k}|^2$
    ,
    $|\tilde{g}_{\tilde{m}n_{\tilde{m}}l}|^2$
    ,
    $|\tilde{g}_{{\tilde{m}n_{\tilde{m}}}\ddot{l}}|^2$
    and
    $|\tilde{g}_{{\tilde{m}n_{\tilde{m}}}\Breve{l}}|^2$ are independent rvs with ${\tilde{m}}=1, \dots, \tilde{M}$, $n_{\tilde{m}}=1, \dots, \mathcal{N}_{\tilde{m}}$ such that $1\leq \mathcal{N}_{\tilde{m}}\leq M-(\tilde{M}-1)$ and $\{k, l, \ddot{l}, \Breve{l}\}=1, \dots, K$ such that $k\neq l\neq \ddot{l}\neq \Breve{l}$. The  probability of outage $P_{\text{out}}^{(k)}$ of the $k$-th user in the proposed dynamic cell-fee network under random AP clustering is identical to that given in \textbf{Theorem 1} with 
    \[
    \Dot{\alpha}_{\tilde{m}n_{\tilde{m}}k}=
    \frac{
    \left(
    \sum_{\tilde{m}=1}^{\tilde{M}}
    \sum_{n_m}^{\mathcal{N}_{\tilde{m}}}
    \tilde{\alpha}_{\tilde{m}n_{\tilde{m}}k}/ \tilde{\beta}_{\tilde{m}n_{\tilde{m}}k}
    \right)^2
    }
    {\sum_{\tilde{m}=1}^{\tilde{M}}
    \sum_{n_{\tilde{m}}=1}^{\mathcal{N}_{\tilde{m}}}
    \tilde{\alpha}_{\tilde{m}n_{\tilde{m}}k}/\tilde{\beta}_{\tilde{m}n_{\tilde{m}}k}^2
    },~~~~~~
    \Dot{\beta}_{\tilde{m}n_{\tilde{m}}k}=\frac{\sum_{\tilde{m}=1}^{\tilde{M}}
\sum_{n_{\tilde{m}}=1}^{\mathcal{N}_{\tilde{m}}}
\tilde{\alpha}_{\tilde{m}n_{\tilde{m}}k}/\tilde{\beta}_{\tilde{m}n_{\tilde{m}}k}^2}{\sum_{\tilde{m}=1}^{\tilde{M}}
\sum_{n_{\tilde{m}}=1}^{\mathcal{N}_{\tilde{m}}}
\tilde{\alpha}_{\tilde{m}n_{\tilde{m}}k}/ \tilde{\beta}_{\tilde{m}n_{\tilde{m}}k}},
\]
\[
\Dot{\alpha}_{\tilde{m}n_{\tilde{m}}k'}=\frac{
\left(
\sum_{\tilde{m}=1}^{\tilde{M}}
\sum_{n_{\tilde{m}}=1}^{\mathcal{N}_{\tilde{m}}}
    \left[
    \sum_{l=1, l\neq k}^K
    \frac{    \tilde{\alpha}_{\tilde{m}n_{\tilde{m}}l}}{ \tilde{\beta}_{\tilde{m}n_{\tilde{m}}l}}
    +
    \sum_{\dot{l}=1, \dot{l}\neq k}^K
    \sum_{\ddot{l}=1, \ddot{l}\neq \dot{l}}^K
    \frac{\tilde{\alpha}_{{\tilde{m}n_{\tilde{m}}}\ddot{l}}}{     \tilde{\beta}_{{\tilde{m}n_{\tilde{m}}}\ddot{l}}}
    +
    \sum_{\Breve{l}=1, \Breve{l}\neq k}^K
    \frac{\tilde{\alpha}_{{\tilde{m}n_{\tilde{m}}}\Breve{l}}}{        \tilde{\beta}_{{\tilde{m}n_{\tilde{m}}}\Breve{l}}}
    \right]
\right)^2
}
{
\sum_{\tilde{m}=1}^{\tilde{M}}
    \sum_{n_{\tilde{m}}=1}^{\mathcal{N}_{\tilde{m}}}
    \left[
    \sum_{l=1, l\neq k}^K
    \frac{\tilde{\alpha}_{\tilde{m}n_{\tilde{m}}l}}{    \tilde{\beta}_{\tilde{m}n_{\tilde{m}}l}^2}
    +
    \sum_{\dot{l}=1, \dot{l}\neq k}^K
    \sum_{\ddot{l}=1, \ddot{l}\neq \dot{l}}^K
    \frac
    {
    \tilde{\alpha}_{{\tilde{m}n_{\tilde{m}}}\ddot{l}}
    }
    {
    \tilde{\beta}_{{\tilde{m}n_{\tilde{m}}}\ddot{l}}^2
    }
    +
    \sum_{\Breve{l}=1, \Breve{l}\neq k}^K
    \frac{
    \tilde{\alpha}_{{\tilde{m}n_{\tilde{m}}}\Breve{l}}
    }
    {
    \tilde{\beta}_{{\tilde{m}n_{\tilde{m}}}\Breve{l}}^2
    }
    \right]
},
\]
\[
\Dot{\beta}_{\tilde{m}n_{\tilde{m}}k'}=
\frac{
\sum_{\tilde{m}=1}^{\tilde{M}}
\sum_{n_{\tilde{m}}=1}^{\mathcal{N}_{\tilde{m}}}
    \left[
    \sum_{l=1, l\neq k}^K
    \frac{\tilde{\alpha}_{\tilde{m}n_{\tilde{m}}l}}{    \tilde{\beta}_{\tilde{m}n_{\tilde{m}}l}^2}
    +
    \sum_{\dot{l}=1, \dot{l}\neq k}^K
    \sum_{\ddot{l}=1, \ddot{l}\neq \dot{l}}^K
    \frac{\tilde{\alpha}_{{\tilde{m}n_{\tilde{m}}}\ddot{l}}}{
    \tilde{\beta}_{m\ddot{l}}^2}
    +
    \sum_{\Breve{l}=1, \Breve{l}\neq k}^K
    \frac{    \tilde{\alpha}_{{\tilde{m}n_{\tilde{m}}}\Breve{l}}}{\tilde{\beta}_{{\tilde{m}n_{\tilde{m}}}\Breve{l}}^2}
    \right]
}
{
 \sum_{\tilde{m}=1}^{\tilde{M}}
    \sum_{n_{\tilde{m}}=1}^{\mathcal{N}_{\tilde{m}}}
    \left[
    \sum_{l=1, l\neq k}^K
    \frac{\tilde{\alpha}_{\tilde{m}n_{\tilde{m}}l}}{     \tilde{\beta}_{\tilde{m}n_{\tilde{m}}l}}
    +
    \sum_{\dot{l}=1, \dot{l}\neq k}^K
    \sum_{\ddot{l}=1, \ddot{l}\neq \dot{l}}^K
    \frac{\tilde{\alpha}_{{\tilde{m}n_{\tilde{m}}}\ddot{l}
}}{ \tilde{\beta}_{{\tilde{m}n_{\tilde{m}}}\ddot{l}}}
    +
    \sum_{\Breve{l}=1, \Breve{l}\neq k}^K
    \frac{\tilde{\alpha}_{{\tilde{m}n_{\tilde{m}}}\Breve{l}
}}{\tilde{\beta}_{{\tilde{m}n_{\tilde{m}}}\Breve{l}}}
    \right]
}
.
\]
	\end{corollary}
    \begin{proof}
This can be easily concluded due to the fact that equations (\ref{SINR_1}) and (\ref{SINR_2}) are different from each other only by the scaling parameters of the rvs.
\end{proof}
Note that using the AP clustering scheme with DAS at every virtual AP, there will be no requirements for CSI from other virtual APs. Additionally, a significant reduction of the lengths of the beamforming vectors will be achieved with a direct positive effect on processing speed and complexity  for the beamforming optimization.
It is important to mention that the dynamic cell-free network model is a generalization of the static cell-free network model. 
This can be inferred by setting $\tilde{M}=M$ which will result in $\mathcal{N}_{\tilde{m}}=1,~\forall \tilde{m}=1, \dots, \tilde{M}(\equiv M)$. Here,  $G_{\tilde{m}n_{\tilde{m}k}}$ will be merged with $w_{\tilde{m}n_{\tilde{m}}k}(\equiv w_{mk})$.  
%===========================================================================================
\section{Clustering Technique and Beamforming Design}
This section formulates the  general problem of jointly performing the clustering of APs and designing the beamforming vectors.
However, before presenting the problem formulation, we describe the SIC technique for  detection of the signals received by the CPU and also the diversity technique to combine the signals from multiple antennas. SIC can be achieved by utilizing the fact that the signals from all users are superimposed in the power domain, and with the proper beamforming design, the signal related to every user can have a distinct power level. 
%=======================================================================
\subsection{SIC-Enabled Signal Detection}
Note that in a cell-free network, signals from different users use the same time-frequency resource. The main idea of power-domain SIC technique is that a proper power control is performed such that signals from different users have a distinct power levels\footnote{This is similar to traditional power-domain non-orthogonal multiple access (NOMA), where signals from different users are multiplexed in the same time-frequency resource and SIC techniques are used at the receivers.}. 
Additionally, to detect signal from a certain user, the SIC unit first detects signals of users with higher power levels, subtracts their contributions from the overall received signal, and then detects the intended signal.
Specifically, for the proposed dynamic cell-free network, when detecting the signal from the $k$-th user, the uplink beamforming vectors are designed such that $p_1\sum_{\tilde{m}=1}^{\tilde{M}}~
    \sum_{n_{\tilde{m}}=1}^{\mathcal{N}_{\tilde{m}}}G_{\tilde{m}n_{\tilde{m}}k}^2
    |\hat{g}_{\tilde{m}n_{\tilde{m}}1}|^2 
    \leq \dots \leq
    p_k\sum_{\tilde{m}=1}^{\tilde{M}} 
    \sum_{n_{\tilde{m}}=1}^{\mathcal{N}_{\tilde{m}}}G_{\tilde{m}n_{\tilde{m}}k}^2
    |\hat{g}_{\tilde{m}n_{\tilde{m}}k}|^2~~
    \leq
    \dots
    \leq
p_K\sum_{\tilde{m}=1}^{\tilde{M}} 
    \sum_{n_{\tilde{m}}=1}^{\mathcal{N}_{\tilde{m}}}G_{\tilde{m}n_{\tilde{m}}k}^2
    |\hat{g}_{\tilde{m}n_{\tilde{m}}K}|^2$.
    
Accordingly, the SINR $\gamma_k$ of the $k$-th user can be expressed as
\begin{dmath}
    \gamma_k^{\{\mathcal{C}_n\}}=
%    \frac
%    {
%    \sum_{\tilde{m}=1}^{\tilde{M}}w_{\tilde{m}k}^2
%   \sum_{n_{\tilde{m}}=1}^{\mathcal{N}_{\tilde{m}}}
%   p_kG_{\tilde{m}n_{\tilde{m}}k}^2
%    |\hat{g}_{\tilde{m}n_{\tilde{m}}k}|^2
%    }
%    {
%    \sum_{\tilde{m}=1}^{\tilde{M}}w_{\tilde{\tilde{m}}k}^2
%    \sum_{n_{\tilde{m}}=1}^{\mathcal{N}_{\tilde{m}}}|G_{\tilde{m}n_{\tilde{m}}k}|^2
%   \sum_{l=1, l\neq k}^K p_l |\hat{g}_{\tilde{m}n_{\tilde{m}}l}|^2
%    +
%   \sum_{\tilde{m}=1}^{\tilde{M}}
%w_{\tilde{m}k}^2
% \sum_{n_{\tilde{m}}=1}^{\mathcal{N}_{\tilde{m}}}G_{\tilde{m}n_{\tilde{m}}k}^2
%  \tilde{\sigma}^2_{\tilde{m}n_{\tilde{m}}}/2
%    }
%    \\
%    =
    \frac
    {
    \sum_{\tilde{m}=1}^{\tilde{M}}
    w_{\tilde{m}k}^2
    \sum_{n_{\bar{m}}=1}^{\mathcal{N}_{\bar{m}}}
    | \check{g}_{\tilde{m}n_{\tilde{m}}k}|^2
    }
    {
    \sum_{\tilde{m}=1}^{\tilde{M}}
    w_{\tilde{m}k}^2
    \sum_{n_{\tilde{m}}=1}^{\mathcal{N}_{\tilde{m}}}
    \left[
    \sum_{l=1, l\neq k}^K|\check{g}_{\tilde{m}n_{\tilde{m}}l}|^2
    +
    \sum_{\dot{l}=1, \dot{l}\neq k}^K
    \sum_{\ddot{l}=1, \ddot{l}\neq \dot{l}}^K
    |\check{g}_{{\tilde{m}n_{\tilde{m}}}{\ddot{l}}}|^2
    +
    \sum_{\Breve{l}=1}^{k-1}|\check{g}_{{\tilde{m}n_{\tilde{m}}}{\Breve{l}}}|^2
    \right]+1
    } 
,\label{SINR_3}
\end{dmath}
where $|\Breve{g}_{\tilde{m}n_{\tilde{m}}k}|^2=\frac{
    G_{\tilde{m}n_{\tilde{m}}k}\tau_p\rho_kp_k\mathcal{E}_{\tilde{m}n_{\tilde{m}}k}^2}
    {   
 \dot{\sigma}_{\tilde{m}n_{\tilde{m}}k} }|{g}_{\tilde{m}n_{\tilde{m}}k}|^2$, 
$|\Breve{g}_{\tilde{m}n_{\tilde{m}}l}|^2
=\frac{
   G_{\tilde{m}n_{\tilde{m}}k}\tau_p\rho_kp_k\mathcal{E}_{\tilde{m}n_{\tilde{m}}k}^2|\bm{\varphi}_{\tilde{m}n_{\tilde{m}}k}^H\bm{\varphi}_{\tilde{m}n_{\tilde{m}}l}|^2}
    {
     \dot{\sigma}_{\tilde{m}n_{\tilde{m}}l} 
    }|{g}_{\tilde{m}n_{\tilde{m}}l}|^2
$,\\
$
|\Breve{g}_{{\tilde{m}n_{\tilde{m}}}{\ddot{l}}}|^2
=
\frac{
G_{\tilde{m}n_{\tilde{m}}k}\tau_p\rho_{\dot{l}}p_{\dot{l}}\mathcal{E}_{{{\tilde{m}n_{\tilde{m}}}}{\dot{l}}}^2|\bm{\varphi}_{{{\tilde{m}n_{\tilde{m}}}}\dot{l}}^H\bm{\varphi}_{{{\tilde{m}n_{\tilde{m}}}}\ddot{l}}|^2
}
{
 \dot{\sigma}_{{\tilde{m}n_{\tilde{m}}}k}
}|{g}_{{\tilde{m}n_{\tilde{m}}}{\ddot{l}}}|^2
$, 
and 
$
|\Breve{g}_{{\tilde{m}n_{\tilde{m}}}{\Breve{l}}}|^2
=
\frac{
\tau_p\rho_{\Breve{l}}p_{\dot{l}}\mathcal{E}_{{{\tilde{m}n_{\tilde{m}}}}{\dot{l}}}^2
}
{
 \dot{\sigma}_{{\tilde{m}n_{\tilde{m}}}k}
}|{g}_{{\tilde{m}n_{\tilde{m}}}{\Breve{l}}}|^2
$.
Note that     
    $ |{g}_{\tilde{m}n_{\tilde{m}}k}|^2$
    , 
    $ |{g}_{\tilde{m}n_{\tilde{m}}l}|^2 $
    , 
    $ |{g}_{{\tilde{m}n_{\tilde{m}}}\ddot{l}}|^2$, 
    and 
    $ |{g}_{{\tilde{m}n_{\tilde{m}}}\Breve{l}}|^2$
represent different order statistical rvs of gamma rvs which prevents the direct utilization of \textit{Welch-Satterthwaite} approximation as in \textbf{Appendix A}. 
This is due to the introduction of random variable ordering required for SIC-based detection which requires the pre-ordering of different users' signals based on their different power levels. However, a closed-form expression for the outage probability can be derived using order statistics theory as in \cite{GCoMP}.
%which was not done here due to space limitations on this work.
We can also notice from (\ref{SINR_3}) that the CSI estimation error contains portions from all users including those with higher power levels than that of the $k$-th user. This is because the SIC operation is also affected by channel estimation errors \cite{SIC_IMPERFECT_CSI}.
%======================================================================= 
\subsection{Diversity Combining Scheme}
Several diversity techniques can be used at multi-antenna transmitters/receivers to combat multipath fading. However, maximal ratio combining (MRC) and selection combining (SC) are the most common diversity combining techniques \cite{simonmk05}.
In this work, we propose the use of a modified version of \textit{Wiener-Hopf} multiple antennas combining scheme under the existence of co-channel interference \cite{iet:/content/books/ew/sbew046e}.
Accordingly, the gain factor $\bm{G}_{\tilde{m}k}=\left[G_{\tilde{m}1k} \dots G_{\tilde{m}{\mathcal{N}}_{\tilde{m}}k}\right]^H$ that is used in combining signals at the $\tilde{m}$-th AP for the detection of $k$-th user signal is defined as \cite{iet:/content/books/ew/sbew046e}
\begin{equation}
    \bm{G}_{\tilde{m}k}=\bm{R}_{\tilde{m}}^{-1}\hat{\bm g}_{\tilde{m}k},
\end{equation}
where $\bm{R}_{\tilde{m}}=\text{Cov}\left(\sum_{l=1, l\neq k}^K\sqrt{{p}_l}\bm{\hat g}_{\tilde{m}l}+\tilde{\bm \eta}_{\tilde{m}}\right)$,   $\bm{\hat g}_{\tilde{m}l}=\left[{\hat g}_{\tilde{m}1l} \dots {\hat g}_{\tilde{m}\mathcal{N}_{\tilde{m}} l} \right]^H$, 
and
$\tilde{\bm \eta}_{\tilde{m}}=\left[\tilde{ \eta}_{\tilde{m}1} \dots \tilde{ \eta}_{\tilde{m}\mathcal{N}_{\tilde{m}}} \right]^H$.  
Furthermore, the covariance matrix $\bm{R}_{\tilde{m}}$ conditioned on instantaneous CSI can be written as
\begin{equation}
\bm{R}_{\tilde{m}}=\sum_{l=1, l\neq k}^K p_l\hat{\bm g}_{\tilde{m}l}\hat{\bm g}_{\tilde{m}l}^H+ \bm{\tilde{\Sigma}}_{\tilde{m}},
\end{equation}
where $\bm{\tilde{\Sigma}}_{\tilde{m}}\in \mathbb{C}^{\mathcal{N}_{\tilde{m}}\times \mathcal{N}_{\tilde{m}}}$ is a diagonal matrix with $n_{\tilde{m}}$-th diagonal element given by ${\tilde{\Sigma}}_{\tilde{m}n_{\Tilde{m}}}=\tilde{\sigma}_{\tilde{m}n_{\tilde{m}}}^2/2$.
Note that, when $K=1$, we have $G_{\tilde{m}k}=2\hat{g}_{\tilde{m}k}/\tilde{\sigma}_{\tilde{m}n_{\tilde{m}}}$ which is identical to the MRC scheme that represents the optimal combining scheme under no IUI. 
Additionally, note that we use the estimated values of CSI ($\hat{g}_{\tilde{m}n_{\tilde{m}}k}$) without considering the CSI estimation error. However, several works in the literature discussed the optimal diversity combining scheme under imperfect CSI estimation \cite{DAS_IMPERFECT_CSI_1,DAS_IMPERFECT_CSI_2}, which is however not investigated in this paper for brevity.  
%======================================================================= 
\subsection{Formulation of Optimization Problem }
For the proposed dynamic cell-free network, the process of AP clustering needs to be performed jointly with beamforming design.
This can be achieved by forming a general optimization problem that jointly finds the optimal cluster set-beamforming vector pairs such that a certain objective is satisfied.
%Here, we assume that the number of clusters $\bar{M}\leq M$ is constant design parameter.
To achieve the best gain from SIC detection, we propose to order users based on their overall weighted power gain of the communication channel such that
$p_1\sum_{\tilde{m}=1}^{\tilde{M}}~
    \sum_{n_{\tilde{m}}=1}^{\mathcal{N}_{\tilde{m}}}G_{\tilde{m}n_{\tilde{m}}k}^2
    |\hat{g}_{\tilde{m}n_{\tilde{m}}1}|^2 
    \leq \dots \leq
    p_k\sum_{\tilde{m}=1}^{\tilde{M}} 
    \sum_{n_{\tilde{m}}=1}^{\mathcal{N}_{\tilde{m}}}G_{\tilde{m}n_{\tilde{m}}k}^2
    |\hat{g}_{\tilde{m}n_{\tilde{m}}k}|^2~~
    \leq
    \dots
    \leq
p_K\sum_{\tilde{m}=1}^{\tilde{M}} 
    \sum_{n_{\tilde{m}}=1}^{\mathcal{N}_{\tilde{m}}}G_{\tilde{m}n_{\tilde{m}}k}^2
    |\hat{g}_{\tilde{m}n_{\tilde{m}}K}|^2$.
Accordingly, the general problem of AP clustering and beamforming optimization can be formulated as
\begin{equation}
\begin{aligned}
& ~\textbf{P}_1: \underset{n, \bm{W}}{\text{max}}~
%\underset{k=1, \dots K}{\text{min}}
& \text{\hspace{-40mm}} \sum_{k=1}^K\log_2\left(1+\gamma_{k}^{\{\mathcal{C}_n\}} \right)\\
& ~\text{Subject to:} \\
& ~\textbf{C}_1: \sum_{m=1}^{M} \left(w_{m \delta_l}^2-\sum_{i=\delta_l+1}^{l}w_{mi}^2 \right)\bar{\gamma}_{\tilde{m}n_{\tilde{m}}l}\geq P_{\text{s}}, \\
& ~  \textbf{C}_2: 0\leq w_{\tilde{m}k}\leq 1 ,
\\
&  ~  \textbf{C}_3:||\bm{w}_{k}||^2\leq 1,
\\
&~~\forall k,  \forall~\delta_l= 1, \dots, l-1\;\mbox{and}\; l=2, \dots, K,\\
\end{aligned}\label{OptimizationProb}
\end{equation}
where $\bar{\gamma}_{\tilde{m}n_{\tilde{m}}l}=\frac{2p_lG_{\tilde{m}n_{\tilde{m}l}}}{\tilde{\sigma}_{\tilde{m}n_{\tilde{m}}}^2}|\hat{g}_{\tilde{m}n_{\tilde{m}l}}|^2$. 
The constraints $\textbf{C}_1$ refer to the set of $\sum_{l=2}^K(l-1)=\frac{K(K-1)}{2}$ conditions required for successful SIC operation with receiver sensitivity of $P_{\text{s}}$. For maximizing the minimum rate, the objective function of the optimization problem will be replaced by $\underset{k}{\text{min}} \log_2\left(1+\gamma_{k}^{\{\mathcal{C}_n\}} \right)$.
Note that the receiver deals only with measured channel values (estimated) which include the estimation error as well as the AWGN part. However, the achieved SINR value after the SIC procedure will be decreased by pilot contamination components (as can be noticed from the second part of Eq. (\ref{SINR_3})).  

To optimally solve the problem $\textbf{P}_1$ in (\ref{OptimizationProb}), 
we need to simultaneously solve for $\mathcal{C}_n$ and $\bm{W}$.
Specifically, for every possible clustering configuration, there is a related optimal beamforming vector that maximizes the objective function.
{\em The globally optimal solution is then the one that gives the best performance among all clustering configurations and their corresponding optimized beamforming vectors ($\bm{W}$).
Such a problem belongs to the set of complete non-deterministic polynomial-time hard (NP-hard) problems.
%This is due to the fact that while it is easy to confirm whether a proposed solution is valid, it may inherently be prohibitively difficult to determine in the first place whether any solution exists. 
Furthermore, given that a certain clustering configuration is selected, solving $\textbf{P}_1$ w.r.t $\bm{W}$ grows exponentially with $\tilde{M}$ and/or $K$ which makes the cell-free network model  impractical for a network with a massive number of devices communicate through a massive number of APs within a  geographic area.
Therefore, for  practical implementation of a dynamic cell-free network, we design a novel hybrid DDPG-DDQN-based DRL system that produces the best clustering configuration and the beamforming vectors such that the objective function in $\textbf{P}_1$ is maximized.}
%=============================================================
%\subsection{Complexity Analysis}

%======================================================================= 

%===========================================================================================
\section{DRL Implementation of AP Clustering and Beamforming Design}

In this section, we introduce a novel DRL model that solves the general optimization problem (i.e. problem $\textbf{P}_1$) of jointly optimizes the clustering of APs and the beamforming vectors. 

\subsection{Theoretical Preliminaries}
The concept of reinforcement learning (RL) refers to the learning process of an agent interacting with its environment after receiving certain observations. The environment provides a reward to the agent for every interaction and
the RL agent aims to select the right action for the next interaction in order to maximize the discounted reward over a time horizon. 
This problem can be formulated as a \textit{Markov decision
process} (MDP). An MDP is a tuple ($\mathcal{\bm S}$, $\mathcal{A}$, $\mathcal{P}$, $\mathcal{R}$, $\zeta$), where $\mathcal{S}$ represents the states space which contains the set of $K$-dimensional states (with each state at time $t$ denoted by $\bm{s}_t$), $\mathcal{A}$ is the action space that contains a finite set of actions from which the agent can choose, $\mathcal{P}$ : $\mathcal{S}$ $\times$ $\mathcal{A}$ $\times$ $\mathcal{S}$ $\rightarrow$ [0, 1] is a transition probability in which $\mathcal{P}(\bm{s}, a, \bm{s}')$ defines the probability of observing state $\bm{s}'$ after executing action $a$ in the state $\bm{s}$, $\mathcal{R}: \mathcal{S} \times \mathcal{A} \rightarrow \mathbb{R}$ is the expected reward after being in state $\bm{s}$ and taking action $a$, and $\zeta$ $\in [0, 1)$ is the discount factor.
To solve the MDP,  RL algorithms have been developed to learn and find a discrete value function or a ``policy". Such a discretization can lead to lack of generalization and significantly increase the problem's dimensionality. 
Therefore, deep RL (DRL) algorithms based on function approximation by deep neural networks (DNNs) have been proposed.

DRL algorithms can be classified into three types: (i) \textit{value-based} methods such as  deep Q-learning (DQL) and SARSA which only learn the so-called  value function to find a policy, (ii) \textit{policy-based} methods which learn the policy directly by following the gradient with respect to the policy, and (iii) \textit{actor-critic} methods which are a hybrid of the value-based for the critic and policy-based methods for the actor.
\iffalse
\cite{DRL_1}:
\begin{itemize}
    \item \textit{Value-based} methods such as  deep Q-learning (DQL) and SARSA which only learn the so-called  value function to find a policy.
    \item \textit{Policy-based} methods which learn the policy directly by following the gradient with respect to the policy itself in  order to keep improving the policy constantly. These methods suffer from noisy gradients and high variance.
    \item \textit{Actor-critic} methods which are a hybrid of the value-based methods for the critic and the policy-based methods for the actor. The critic reduces the variance of policy gradient methods by estimating the action-value function $Q(\bm{s},a)$.
\end{itemize}
\fi

The standard DQL method represents the most popular update algorithm in the literature due to the availability of a mature theory. Basically, the DQL update equation at time $t$ for a network agent with parameters $\theta^Q$ after taking action $a_t$ in state $\bm{s}_t$ and observing the
immediate reward $r_{t+1}$ and resulting state $\bm{s}_{t+1}$ is:
\begin{equation}
\label{eq:dqn-update equation}
\begin{aligned}
Q(\bm{s}, a \,|\, \theta^Q_{t+1}) &= Q(\bm{s}, a\,|\, \theta^Q_t) + \nu \bigg[r_{t+1} + \zeta  \max_{a'} Q(\bm{s}_{t+1}, a'\,|\, \theta^Q_t) - Q(\bm{s}_t, a_t\,|\, \theta^Q_t)\bigg]\\
&\text{\hspace{-30mm}}= Q(\bm{s}, a\,|\, \theta^Q_t) + \nu \bigg[r_{t+1} + \zeta \max_{a'} Q(\bm{s}_{t+1}, \argmax_{a'} Q(\bm{s}_{t+1},a'\,|\, \theta^Q_t)\,|\, \theta^Q_t) - Q(\bm{s}_t, a_t\,|\, \theta^Q_t)\bigg],
\end{aligned}
\end{equation}
where $\nu$ is the learning rate.
Computing the term $ \underset{a'}{\text{max}}~ Q(\bm{s}_{t+1}, \underset{a'}{\text{argmax}}  Q(\bm{s}_{t+1},a'\,|\, {\theta}^Q_t)\,|\, \theta^Q_t)$ introduces a systematic overestimation of the Q-values during the learning that is accentuated by the use of bootstrapping, i.e. learning estimates from estimates. The Q-learning update in (\ref{eq:dqn-update equation}) uses the same Q-network $Q(\bm{s}, a\,|\, {\theta}^Q_t)$ both to select and to evaluate an action. After highlighting the overestimation bias in experiments across different Atari game environments,  Hasselt et al. \cite{van2016deep} decoupled the action selection and evaluation by introducing two deep Q-networks,  a  $Q$ network and a target network $Q'$ with different parameters $\theta^Q$ and ${\theta}^{Q'}$, respectively, to avoid the maximization bias.
The  $Q'$ network is used for action selection while the $Q$ network is used for action evaluation.
This is known as {\em deep double Q-learning algorithm (DDQL)}.
The DDQL update equation of the network can be expressed as:
\begin{dmath}
\text{\hspace{-2mm}}Q(\bm{s}, a \,|\, {\theta}^Q_{t+1}) = Q(\bm{s}, a\,|\, {\theta}^Q_t) 
\\
+ \nu \bigg[r_{t+1} + \zeta \max_{a'} Q(\bm{s}_{t+1}, \argmax_{a'} Q'(\bm{s}_{t+1},a'\,|\, {\theta}^{Q'}_t)\,|\, {\theta}^Q_t) - Q(\bm{s}_t, a_t\,|\, {\theta}^Q_t)\bigg].
\label{eq:ddqn-update-q1-equation}
\end{dmath}
The parameters ${\theta}^{Q'}$ of the  $Q'$ network periodically hard-copy the parameters $\theta^Q$ of $Q$ network after $t_0$ time steps using the Polyak averaging method with parameter $\tau \in [0,1]$:
\begin{equation}
\label{eq:ddqn-update-q2-equation}
{\theta}^{Q'}_{t+t_0} = (1-\tau) \,{\theta}^{Q'}_t + \tau \,\theta^Q_t.
\end{equation}

The DDQL networks show a better performance that standard DQL \cite{van2016deep}; however, due to the discretization requirements of the DNN outputs (the action space $\mathcal{A}$), it results in a huge expansion of the action space dimensionality when used in the optimization of an objective function of continuous dependent variables.
{\em This dimensionality issue makes it an unattractive solution for solving the beamforming problem under massive number of users and APs.}
However, it is a relevant candidate for the clustering problem of the APs since it avoids the need for an extremely inefficient exhaustive search method. 
This  motivates us to utilize the ``DDPG" policy for the beamforming design problem.

DDPG belongs to the class of actor-critic algorithms. It concurrently learns a Q-function network approximation $Q(\bm{s}, a|\theta^{Q})$ called the critic, and a policy network approximation $\mu(\bm{s}|\theta^{\mu})$ called the actor. 
The Q-function network is trained using the Bellman equation, while the policy network is learnt using the Q-function. 
Unlike the DQL policies which output the probability distribution $\pi(a|\bm{s})$ across a discrete action space $\mathcal{A}$,
the policy network of DDPG directly maps states to actions. 
Specifically, at every time step $t$, it maximizes its loss function defined as:
\begin{equation}
\label{eq:loss-function-actor}
J(\theta) = \mathbb{E}\bigg[Q(\bm{s}, a) \;|\; \mathcal{S}=\bm{s}_t, a=\pi(a|\bm{s}_t)\bigg]
\end{equation}
and updates its weights $\theta$ by following the gradient of (\ref{eq:loss-function-actor}):
\begin{equation}
\label{eq:gradient-loss-function-actor}
\nabla J_{\theta^{\mu}}(\theta) \approx \nabla_{a} Q(\bm{s},a) \,\nabla \mu(s|\theta^{\mu}).
\end{equation}
This update rule represents the deterministic policy gradient (DPG) theorem, rigorously proved by Silver et al. in the supplementary material of \cite{silver2014deterministic}.
The term $\nabla_{a} Q(\bm{s},a)$ is obtained from a Q-network  $Q(\bm{s}, a|\theta^{Q})$ called the critic by backpropagating its output w.r.t. the action input $\mu(\bm{s}|\theta^{\mu})$. When the number of actions is very large, this actor-critic training procedure solves the intractability problem of DQN \cite{mnih2015human}  by using the following approximation:
\begin{equation}
\label{eq:max-approx-ddpg}
\max_{a} Q(\bm{s}, a) \approx Q(\bm{s}, a |\theta^{Q})|_{a=\mu(\bm{s}|\theta^{\mu})}.
\end{equation}
Similar to DQN, two tricks are employed to stabilize the training of the DDPG actor-critic architecture, namely, 1) the experience replay buffer $R$ to train the critic, and 2) target networks for both the actor and the critic which are updated using the polyak averaging in the same way it was done in (\ref{eq:ddqn-update-q2-equation}).
Now that we have provided a brief DRL background on the two methods used by our proposed algorithm (DDQL and DDPG), 
a detailed description of the proposed DRL-based AP clustering and beamforming design for a dynamic cell-free network can be presented in the following subsection.
%=========================================================================
\subsection{DRL Agent Design for AP Clustering and Beamforming Optimization}
Our goal is to design a DRL system that jointly optimizes the clustering of APs and the beamforming vectors given a certain CSI matrix
$\bm{H}=
\left[
{\tilde{\bm{g}}}_{1}
\dots 
{\tilde{\bm{g}}}_{K}
\right]$.
In this context, we develop a hybrid DDPG-DDQL DRL scheme that simultaneously learns the best clustering subsets-beamforming vector pairs given a certain CSI matrix $\bm{H}$ (or any other metric related to $\bm{H}$).
Fig. \ref{AI_Model} shows a schematic block diagram of the developed hybrid DDPG-DDQL DRL scheme.
   	\begin{figure}[htb]
		\centering
		\includegraphics[scale=0.35]{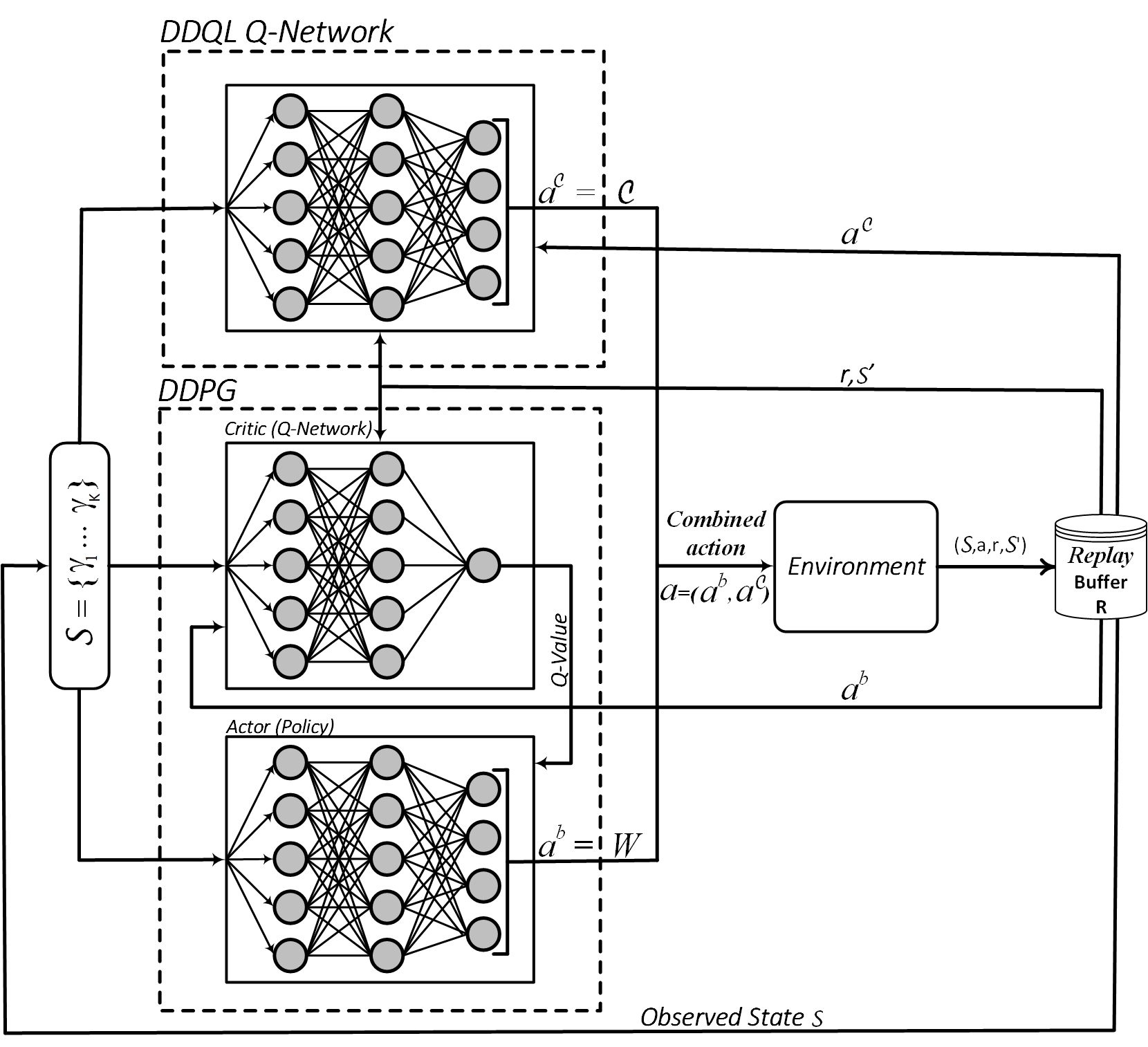}
		\caption{Hybrid DDPG-DDQL model for AP clustering and uplink beamforming design. }\label{AI_Model}
	\end{figure} 
To find the beamforming decoding matrix of the $k$-th user $\bm{W}_k=\text{Reshap}(\bm{w}_{k})$, where $\bm{W}_k\in \mathbb{R}^{\nint*{\sqrt{\sum_{\tilde{m}=1}^{\tilde{m}}\sum_{n_{\tilde{m}}=1}^{\mathcal{N}_{\tilde{m}}} \bm{w}_{k}}}
\times
\nint*{\sqrt{\sum_{\tilde{m}=1}^{\tilde{m}}\sum_{n_{\tilde{m}}=1}^{\mathcal{N}_{\tilde{m}}} \bm{w}_{k}}}}$\footnote{Note that the empty elements on $\bm{W}_k$ are zero padded.} and  $\bm{w}_{k}=\left[w_{1k} \dots w_{\tilde{M}k}\right]^H$, we opt for the actor-critic DDPG algorithm \cite{lillicrap2015continuous} since all the  elements $w_{i,j}$ of $\bm{W}_k$ are in the continuous range $[0,1]$.
For the AP clustering problem, we use the well-known DDQL algorithm \cite{van2016deep} to find the best access point clustering configuration since the possible number of clustering configurations is finite and the index of each possible configuration is an integer number. 

As shown in Fig.~\ref{AI_Model}, the two algorithms interact with a simulated cell-free network environment to solve the optimization problem $\textbf{P}_1$ in (\ref{OptimizationProb}).
The design of the cell-free network environment involves the specification of the environment state $\bm{s}$ and the definition of immediate reward function $r$ required by the DRL algorithms to approximate the policies and the Q-values. 
The state of the environment is a vector of the SINR values of users, where the state $\bm{s}_k$ of the $k^{th}$-user is chosen to be the SINR value $\gamma_{k}^{\mathcal{C}_n}$ which is a varying function of the instantaneous CSI matrix.
%:\begin{itemize}
%    \item the SINR of the fully centralized uplink precoding problem as defined in Eq. (\ref{SINR_1}).
%    \item the SINR of the uplink precoding problem with clustering as defined in Eq. (\ref{SINR_2}).
%    \item the SINR of the problem of the uplink precoding with clustering with Noma as defined in Eq. (\ref{SINR_3}).
%\end{itemize}
Under this setup, the action space $\mathcal{A}$ is the pair of actions $a=(a^{b}, a^{\mathcal{C}}) = (\bm{W}_k,\mathcal{C}_n)$ that the DDPG and the DDQN output separately. 
The superscripts $b$ and $\mathcal{C}$ refer to  ``\textit{beamforming}" and ``\textit{clustering}", respectively. 
After receiving the environment's state $\bm{s}$, the DDPG algorithm outputs the action $a^{b} = \bm{W}_k$ and the DDQL outputs  the action $a^{\mathcal{C}}$ representing the cluster partition $\mathcal{C}_n$.
%Finally, the reward $r$ follows the fair criteria, i.e. the minimum among all users' rates:
%\begin{equation}
%\label{eq:reward-definition}
%r = \sum_{k=1}^Klog(1+\gamma_{k}^{\mathcal{C}_n}).
%\end{equation}

Table \ref{Table3} summarizes the environment design by specifying the additional problem parameters.
\begin{table}[h!]
    \centering
      \caption{DRL agent design}
\begin{tabular}{|c|c|}
\hline
Environment Variables & System Equivalence  \\
\hline
\hline
State ${\bm{s}}=\{{s}_1, \dots, {s}_K \}$ & $\{\gamma_{1}^{\mathcal{C}_n}, \dots, \gamma_{K}^{\mathcal{C}_n}\}$\\
\hline
Reward $r$ & $ \sum_{k=1}^K  \;\log\left(1+\gamma_{k}^{\mathcal{C}_n}\right)$\\
\hline
Action $\mathcal{A}$ & $(a^{b}, a^{\mathcal{C}}) = (\bm{W}_k, \mathcal{C}_n)$ \\
\hline
$n=1, \dots, N$ & Index of possible Clustering \\
\hline
$\mathcal{C}_n$ & $n$-th clustering configuration \\
\hline
$K$ & Number of active users \\
\hline
\end{tabular}
    \label{Table3}
\end{table}
{\em Note that when the number of clusters equals to the total number of APs ($\tilde{M}=M$), the network will reduce to the static cell-free network with $\mathcal{C}_n=\mathcal{C}, ~  n=1$}. 
Accordingly, only DDPG model will be required to optimize the beamforming vectors ($\bm{W}_k$).\\
Our TensorFlow/Keras implementation of the actor and critic networks (including their corresponding target networks) have two hidden layers of 256 and 128 neurons, respectively. The DDQN networks have two 
fully-connected layers of 64 neurons followed with an activation function \textit{relu} each, and a final linear fully-connected layer. We use a discount factor $\zeta = 0.99$, a learning rate $\nu = 5 \,\cdot\,10^{-5}$, a Polyak averaging parameter $\tau=10^{-3}$, and an experience replay buffer of size $R=20000$. The critic optimizer is Adam with its default hyperparameters $\beta_1=0.9$ and $\beta_2=0.999$. We train all the networks on 2500 episodes, with 500 time step each.
%=====================================================================
\subsection{Description of the Hybrid DDPG-DDQL Algoritm}

For the static cell-free network, only the DDPG algorithm is trained (one clustering configuration). However, for the dynamic cell-free network, both the DDPG and DDQL networks are trained.
\begin{itemize}
    \item \underline{DDPG network}:
    \begin{itemize}
        \item \textit{The actor network $\mu (\bm{s} | \theta^{\mu})$} maps the SINR values of users to the beamforming matrix $\bm{W}_k$. The output of the network is $a^{b}$, i.e. a flatten list of all the elements $w_{ij} \in [0,1]$.
        \item \textit{The target actor network $\mu' (\bm{s} | \theta^{\mu'})$}: time-delayed copy of the actor network $\mu (\bm{s} | \theta^{\mu})$.
        \item \textit{The critic network $Q(\bm{s}, a^b | \theta^{Q})$}: maps the SINR values and the output action of $\mu (\bm{s} | \theta^{\mu})$ to their corresponding Q-value.
        \item \textit{The target critic network $Q'(\bm{s}, a^b | \theta^{Q'})$}: time-delayed copy of the critic network $Q(\bm{s}, a^b | \theta^{Q})$.
    \end{itemize}
    \item \underline{DDQL network}:
    \begin{itemize}
        \item \textit{The $Q_c$-network $Q_c(\bm{s}, a^{\mathcal{C}} | \theta^{Q_c})$} maps the SINR values of users to the the Q-values of the state and all the clustering partitions.
        \item \textit{The target $Q_c$-network $Q_c'(\bm{s}, a^{\mathcal{C}} | \theta^{Q_c'})$}: time-delayed copy of the $Q_c$-network $Q(\bm{s}, a^{\mathcal{C}} | \theta^{Q_c})$.
    \end{itemize}
\end{itemize}

We describe all of the training steps of our algorithm in \textbf{Algorithm \ref{algo:DRL-for-free-cell-network}}. We start by initializing all neural networks and their targets for the beamforming and clustering problems as well as a replay buffer $R$ (lines \ref{algo:init1}--\ref{algo:init5}). For every episode, we initialize the environment with $N$ access points and $K$ users by setting the initial state $\bm{s}_0$ to a random vector of SINR values of size $K$ (line \ref{algo:init-env}). At every time step $t$ of the episode, the DDQL and DDPG agents pick an action $a^{\mathcal{C}}_t$ and $a^b_t$ respectively (lines \ref{algo:pick-ab}--\ref{algo:pick-ac}). The combined action $a_t = (a^b_t, a^{\mathcal{C}}_t)$ is sent to the free-cell network environment which will transit to a new state $\bm{s}_{t+1}$, and this new state will be returned together with the immediate reward $r_t$ (lines \ref{algo:combine-ab-ac}--\ref{algo:get-state-reward}). After storing the transition tuple $(\bm{s}_{t}, a_{t}, r_{t}, s_{t+1})$ in the replay buffer $R$ (line \ref{algo:store-transition}), we randomly sample from the experience replay buffer $N$ transitions to train the DDPG and DDQL networks (line \ref{algo:sample-replay-buffer}). We start the DDPG training in line (\ref{algo:ddpg-td-target}) by computing the TD target for the Q-network $Q(\bm{s}, a | \theta^{Q})$ using the target Q-network $Q'(\bm{s}, a | \theta^{Q'})$. We update the critic $Q(\bm{s}, a | \theta^{Q})$'s parameters $\theta^{Q}$ in line \ref{algo:update-critic} using the gradient of the MSE of the loss function of the TD target and the output of the critic. The update of the actor's parameters $\theta^{\mu}$ uses the Monte Carlo approximation of gradient in line \ref{eq:gradient-loss-function-actor}. 
The target critic and target policy networks get updated slowly every $P$ iterations (lines \ref{algo:update-critic-target}--\ref{algo:update-actor-target}). Finally, we update the parameters $\theta^{Q_c}$ of the DDQL Q-network using the Bellman equation in line (\ref{algo:update-ddql-q-network}) after selecting the action using the target Q-network $Q'_c(\bm{s}, a | \theta^{Q'_c})$ in line \ref{algo:ddql-action-selection}. Similar to the DDPG target network, we update in line \ref{algo:update-ddql-target} the DDQL target Q-network every $P$ iterations.
\begin{algorithm}
\small
\caption{Hybrid DDPG-DDQL algorithm for uplink beamforming and clustering}\label{algo:DRL-for-free-cell-network}
\begin{algorithmic}[1]
\State Randomly initialize the critic $Q(\bm{s}, a | \theta^{Q})$ and the actor $\mu(\bm{s} | \theta^{\mu})$ with weights $\theta^{Q}$ and $\theta^{\mu}$ \label{algo:init1}
\State Initialize target network $Q'$ and $\mu'$ with weights $\theta^{Q^{\prime}} \leftarrow \theta^{Q}$, $\theta^{\mu^{\prime}} \leftarrow \theta^{\mu}$ \label{algo:init2}
\State Randomly initialize the $Q_c$-network $Q_c(\bm{s}, a | \theta^{Q_c})$ \label{algo:init3}
\State Initialize the target network $Q'_c(\bm{s}, a | \theta^{Q'_c})$ with weights $\theta^{Q_c^{\prime}} \leftarrow \theta^{Q_c}$ \label{algo:init4}
\State Initialize replay buffer R \label{algo:init5}
\For {$episode=1,\dots, E$}
\State Receive initial observation state $\bm{s}_1$ after initializing the environment \label{algo:init-env}
\For {$t=1,\dots, T$}
\State  Select the beamforming action $a^{b}_{t}=\mu\left(\bm{s}_{t} | \theta^{\mu}\right)$ \label{algo:pick-ab}
\State  Select the clustering action $a^{\mathcal{C}}_t=\arg \max _{a^{\mathcal{C}}} Q_c\left(\bm{s}_t, a^{\mathcal{C}}\right)$ \label{algo:pick-ac}
\State Define $a_t = (a^b_t, a^{\mathcal{C}}_t)$ \label{algo:combine-ab-ac}
\State Execute action $a_{t}$ and observe reward $r_{t}$ and observe new state $\bm{s}_{t+1}$ \label{algo:get-state-reward}
\State Store transition $(\bm{s}_{t}, a_{t}, r_{t}, \bm{s}_{t+1})$ in $R$ \label{algo:store-transition}
\State Sample a random minibatch of $\mathcal{L}$ transitions $(\bm{s}_{i}, a_{i}, r_{i}, \bm{s}_{i+1})$ from $R$ \label{algo:sample-replay-buffer}
\State Get $a^b_i$ and $a^{\mathcal{C}}_i$ from $a_i$
\Statex\LeftComment{2}{\underline{Training the DDPG networks}}
\State Compute the TD target $y_{i}=r_{i}+\zeta\, Q^{\prime}\left(\bm{s}_{i+1}, \mu^{\prime}\left(\bm{s}_{i+1} | \theta^{\mu^{\prime}}\right) | \theta^{Q^{\prime}}\right)$  \label{algo:ddpg-td-target}
\State Update the critic $Q(\bm{s}, a | \theta^{Q})$ by minimizing the loss: $L=\frac{1}{\mathcal{L}} \sum_{i}\left(y_{i}-Q\left(\bm{s}_{i}, a^b_{i} | \theta^{Q}\right)\right)^{2}$ \label{algo:update-critic}
%\vspace{-7mm}
\Statex \qquad \quad  Update the actor policy $\mu(\bm{s} | \theta^{\mu})$ using a monte-carlo approximation of (\ref{eq:gradient-loss-function-actor}):
\State \qquad \qquad$\left.\left.\nabla_{\theta^{\mu}} J \approx \frac{1}{\mathcal{L}} \sum_{i} \nabla_{a} Q\left(\bm{s}, a | \theta^{Q}\right)\right|_{\mathcal{S}=\bm{s}_{i}, a=\mu\left(\bm{s}_{i}\right)} \nabla_{\theta^{\mu}} \mu\left(\bm{s} | \theta^{\mu}\right)\right|_{\mathcal{S}=\bm{s}_{i}}$
\Statex \qquad \quad Update the DDPG target networks $Q'$ and $\mu'$ \textbf{if} $mod(t, P)=0$:
\State \qquad \qquad ${\theta^{Q^{\prime}} \leftarrow \tau \theta^{Q}+(1-\tau) \theta^{Q^{\prime}}}$ \label{algo:update-critic-target}
\State \qquad \qquad ${\theta^{\mu^{\prime}} \leftarrow \tau \theta^{\mu}+(1-\tau) \theta^{\mu^{\prime}}}$ \label{algo:update-actor-target}
\Statex\LeftComment{2}{\underline{Training the DDQL networks}}
\State select $a^{*}=\arg \max _{a} Q'_c\left(\bm{s}_{i+1}, a| \theta^{Q'_c}\right)$   \label{algo:ddql-action-selection}
\Statex \qquad \quad Update the $Q_c$ usin{}g:
\State \qquad \qquad $Q_c(\bm{s}_i, a^c_i | \theta^{Q_c}) \leftarrow Q_c(\bm{s}_i, a^c_i| \theta^{Q_c})+\nu\,\left(r_i+\zeta\, Q_c\left(s_{i+1}, a^{*}| \theta^{Q_c}\right)-Q_c(\bm{s}_i, a^{\mathcal{C}}_i| \theta^{Q_c})\right)$ \label{algo:update-ddql-q-network}
%\vspace{-7mm}
\Statex \qquad \quad Update the DDQL target networks $Q_c'$ \textbf{if} $mod(t, P)=0$:
\State \qquad \qquad ${\theta^{Q_1^{\prime}} \leftarrow \tau \theta^{Q_1}+(1-\tau) \theta^{Q_1^{\prime}}}$ \label{algo:update-ddql-target}
\EndFor

%\Statex\LeftComment{1}{compute $\mu^t_{i \rightarrow (i,j)}(\mathbf{u}_i)$ and $\nu^t_{j \rightarrow (i,j)}(\mathbf{v}_j)$ $\forall i,j$}
\EndFor
%\EndProcedure
\end{algorithmic}
\end{algorithm}
%%*************************************************************************************************
%\textcolor{red}{@Mohamed: add this section then check with yasser}
%\subsection{Combined access point clustering and decoding scheme using DDPG}
%\textcolor{red}{@Mohamed: add this section then check with yasser}
%http://asl.stanford.edu/wp-content/papercite-data/pdf/Chinchali.ea.AAAI18.pdf
%\newpage
%=============================
%===========================
%===========================
Note that the DRL model can be easily extended to  optimize also the  uplink transmission power for users $P_k, \forall k=1, \dots, K$. However, the power control information will then need to be transmitted to the users, which will increase the signaling traffic significantly.
\section{Numerical Results and Discussions}
%This section provides some insightful results from both analytical and simulation parts of the studied schemes accompanied with different comparisons and discussions under different operating parameters. 

%=======================================================================
\subsection{Parameters and Assumptions}

%We provide the set of parameters used in generating different results.

\begin{table}[h!]
    \centering
      \caption{Simulation parameters}
\begin{tabular}{|c|c|}
\hline 
Parameter & Value  \\
\hline
\hline
AWGN PSD per UE &
$-169$ dBm/Hz  \\
\hline
Path-loss exponent, $\kappa$ & $2$\\
\hline
Nakagami parameters, $(\mathcal{M},\Omega)$ & $(1,1)$\\
\hline
Training sequence length, $\tau_p$ & $K$ Samples\\
\hline
Pilot transmission power, $\rho_k$ & $100$ mW, ${\forall k}$\\
\hline
SIC sensitivity, $P_s$  & $1$ dBm
\\
\hline
\end{tabular}
    \label{Table2}
\end{table}

Table \ref{Table2} presents  the  main  system  parameters  used
to obtain simulation and analytical results. To simplify simulation and analysis of this work and to concentrate on the most insightful conclusions, we establish some operating assumptions related to channel training and CSI fading models.
We assume that users' pilot signals are mutually orthonormal such that $\tau_p=K\leq \tau_c$ with zero AWGN component in CSI estimation, i.e. $|\bm{\varphi}_{lmn}^H\bm{\varphi}_{xyz}|=1$ if $(l,m,n)=(x,y,z)$ or zero otherwise.
$\forall~(l,x)=1, \dots ,\tilde{M}$, $(m,y)=1, \dots, \mathcal{N}_{(l,x)}$ and $(n,z)=1, \dots, K$.
This will result in removal of the CSI estimation error parts (due to both pilot contamination and AWGN component, $\bm{\eta}_{\tilde{m}n_{\tilde{m}}}$) and considering $\hat{g}_{xyz}={g}_{xyz}$.
Additionally, we assume that all channel fading gains $h_{\tilde{m}n_{\tilde{m}}k}$ are independent and identically distributed (i.i.d) with $\mathcal{M}_{\tilde{m}n_{\tilde{m}}k}=\mathcal{M}=1$ and $\Omega_{\tilde{m}n_{\tilde{m}}k}=\Omega=1$.
Similarly, we assume that all AWGN values belong to a set of i.i.d rvs with PSD $\sigma_{\tilde{m}n_{\tilde{m}}}/2=-169$ dBm/Hz. 
As for large-scale fading, we assume that all APs and users are uniformly distributed over a disc of radius 18m (corresponding to a network total coverage area of $1 \text{km}^2$.

First, we  investigate the accuracy of \textit{Welch-Satterthwaite} sum of gamma rvs approximation method 
(Fig. \ref{Welch_Figure_1}).
\begin{figure}[htb]
		\centering
		\includegraphics[scale=0.6]{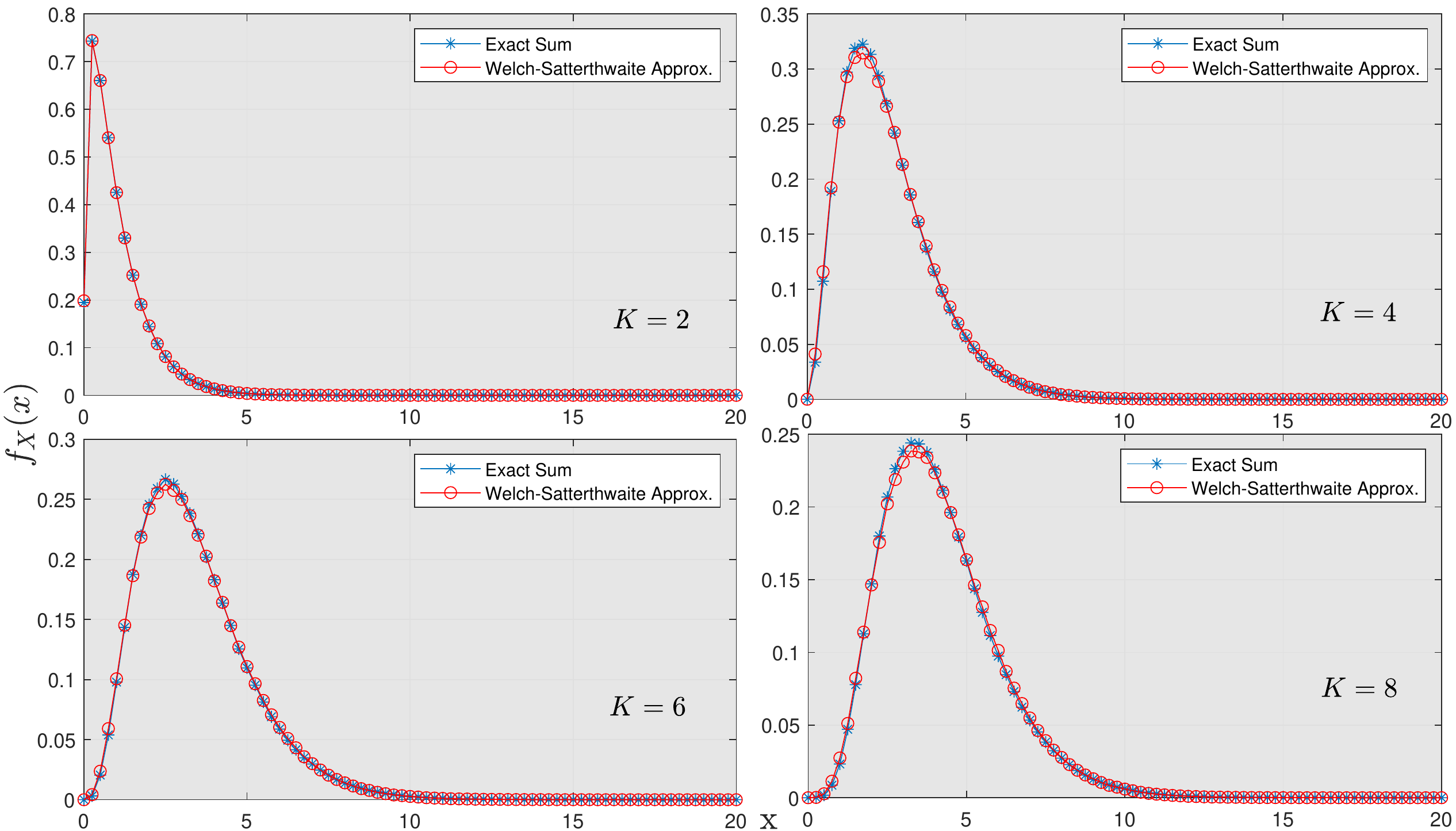}
		\caption{An illustrative example of \textit{Welch-Satterthwaite} approximation with different $\#$ of rvs ($K$).}\label{Welch_Figure_1}
	\end{figure}
	It can be noticed from Fig. \ref{Welch_Figure_1} that a satisfactory accuracy for the sum of non-identically distributed (different $\beta$) gamma rvs can be achieved by \textit{Welch-Satterthwaite} approximation for small, medium, and high numbers of random variables.
	This justifies the use of such an approximation instead of the use of the central limit theorem which requires a relatively large number of random variables.
%=======================================================================
\subsection{Outage Performance}
This section discusses the outage performance of cell-free network in a massive communication regime. 
Each simulation value is obtained via  $2\times10^6$ Monte-Carlo simulation runs. In Fig. \ref{Figure_Outage_DRL_1} (a),
we evaluate the  outage probability for a user in a static cell-free  network under different values of $K$.
\begin{figure}[htb]
		\centering
		\includegraphics[height=9cm, width=14cm]{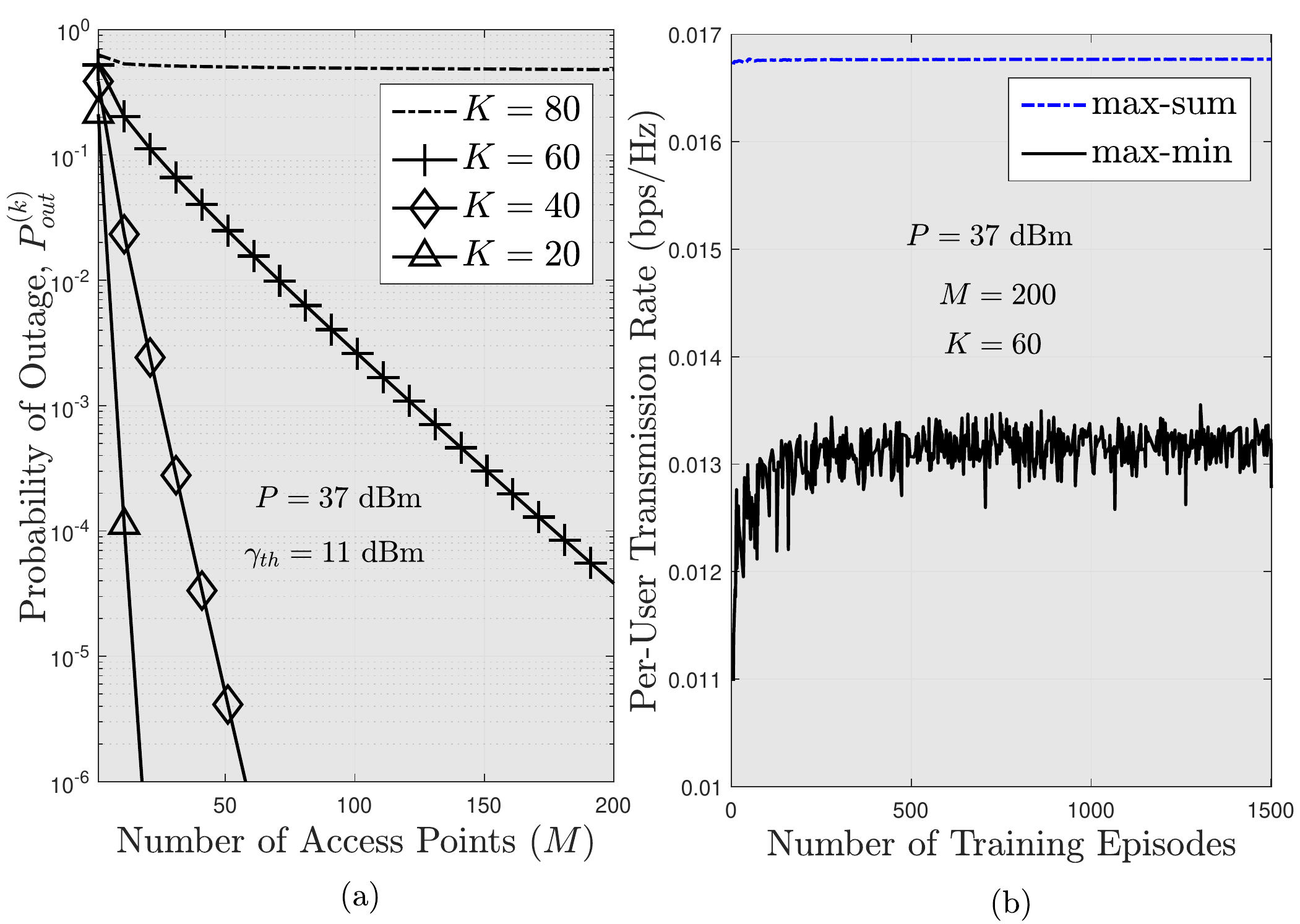}
		\caption{User performance: (a) Probability of outage,  (b) Transmission rate per user. }\label{Figure_Outage_DRL_1}
	\end{figure}
Fig.~\ref{Figure_Outage_DRL_1} (a) shows that an excellent outage performance can be achieved when the number of APs ($M$) is much grater than the number of users $K$, i.e. $M>>K$. 
However, when $K$ becomes comparable to $M$, the outage performance deteriorates significantly with almost total blockage at $K\approx 0.5 M$.
Note that this figure is produced assuming no precoding scheme, i.e. $w_{mk}=c, \forall m, k$. A better performance can be achieved by first solving problem $\textbf{P}_1$ to obtain $\bm{W}_k$ and then substituting $\bm{W}_k$ into (\ref{e_outage_1}).

%Matlab: M150 K50 \rightarrow sum: 0.0888, min: 0.0177
%======================================================
\subsection{DRL Model}
In this section, we study the performance of using the proposed DRL model in solving problem $\textbf{P}_1$.
Fig. \ref{Figure_Outage_DRL_1} (b) shows the per-user transmission rate in the cell-free network with the proposed DRL method.
We can notice that the performance reaches a certain level after around $500$ training episodes and then stays around that level.
Compared with the max-min objective, we notice a significant performance gain with the max-sum objective.
This is due to the fact that, with the max-sum objective, we optimize $\bm{W}_k, \forall k=1, \dots, K$ while with max-min objective, we maximize only $\bm{W}_{k'}$ and use $W_k, k\neq k'~\&~k=1, \dots, K$, where $k'$, denoting the index of the user with minimum achievable rate, is found from previous iterations.

To evaluate the performance of the proposed dynamic cell-free clustering scheme with SIC detection, Fig. \ref{Figure_Two_DRL_1} a shows the per-user transmission rate for both static and dynamic cell-free networks.  
It can be noticed that utilizing SIC at the receiver side significantly compensates for the performance loss caused by clustering of the APs. Note that the per-user transmission rate achieved without SIC in presence of AP clustering is significantly lower than that in a static cell-free network.
 
Additionally, in Fig. \ref{Figure_Two_DRL_1}(b), we compare the performance of the proposed DRL-based solution approach compared to traditional methods.
\begin{figure}[htb]
		\centering
		\includegraphics[height=9cm, width=14cm]{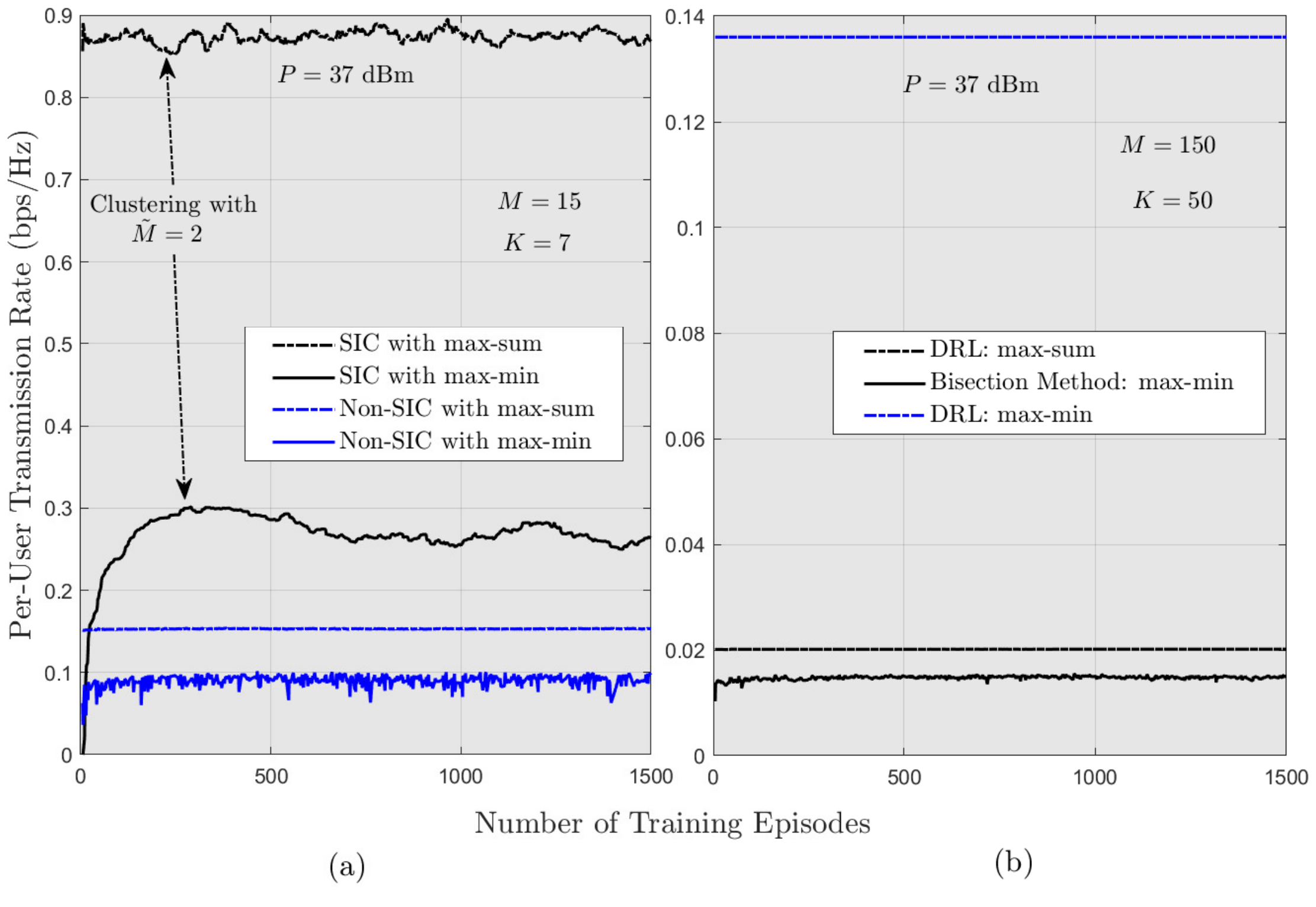}
		\caption{Per-user transmission rate: (a) with SIC, (b) optimal vs. DRL.}\label{Figure_Two_DRL_1}
\end{figure}
We can notice that DRL can achieve per user transmission rate of around $78\%$ of that with conventional bisection method (subgradient method). 
This performance degradation comes with a remarkable decrease in the computational complexity. With the proposed DRL-based design, the beamforming vectors can be obtained in an online manner for practical implementation of non-orthogonal multiple access in cell-free networks.

Figs. \ref{RMSE} (a) and (b) evaluate the convergence rate of the proposed DRL algorithm for two different network setups.
Specifically, we plot the absolute value of the difference between the normalized rate values generated by the DRL model and that by the bisection method (using the Matlab optimization toolkit).
\begin{figure}[htb]
		\centering
		\includegraphics[height=9cm, width=14cm]{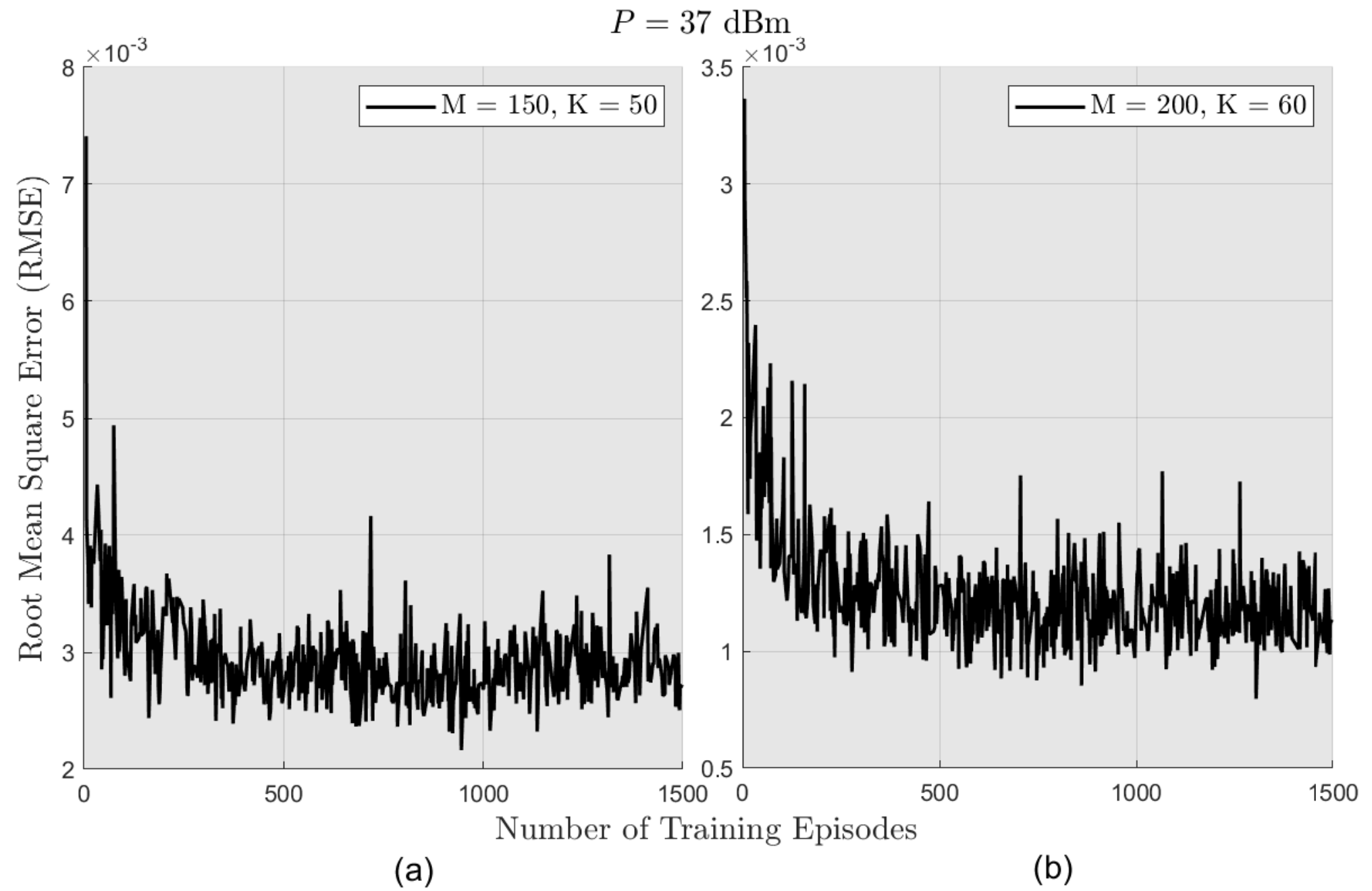}
		\caption{DRL convergence rate: (a) Medium-scale scenario, (b) large-scale scenario.}\label{RMSE}
\end{figure}
It can be noticed that the proposed DRL model shows a better convergence rate than the bisection method under medium-scale scenario with $M = 150$ and $K = 50$ (Fig. \ref{RMSE} [a]) compared to that of large-scale scenario with $M = 200$ and $K = 60$ (Fig. \ref{RMSE} [b]).
However, the proposed DRL model under large-scale scenario is found to achieve a better performance by achieving an approximate of $92\%$ of the normalized rate achieved by bisection method compared to only $78\%$ achieved under medium-scale scenario.   

Finally, to evaluate the performance of the cell-free network under different values of per-user transmission power at the large-scale regime, Fig. \ref{DRL_Vs_Power} illustrates the variation in per-user transmission rate in a large-scale static cell-free network for different transmit power values.
\begin{figure}[htb]
		\centering
		\includegraphics[height=10cm, width=11cm]{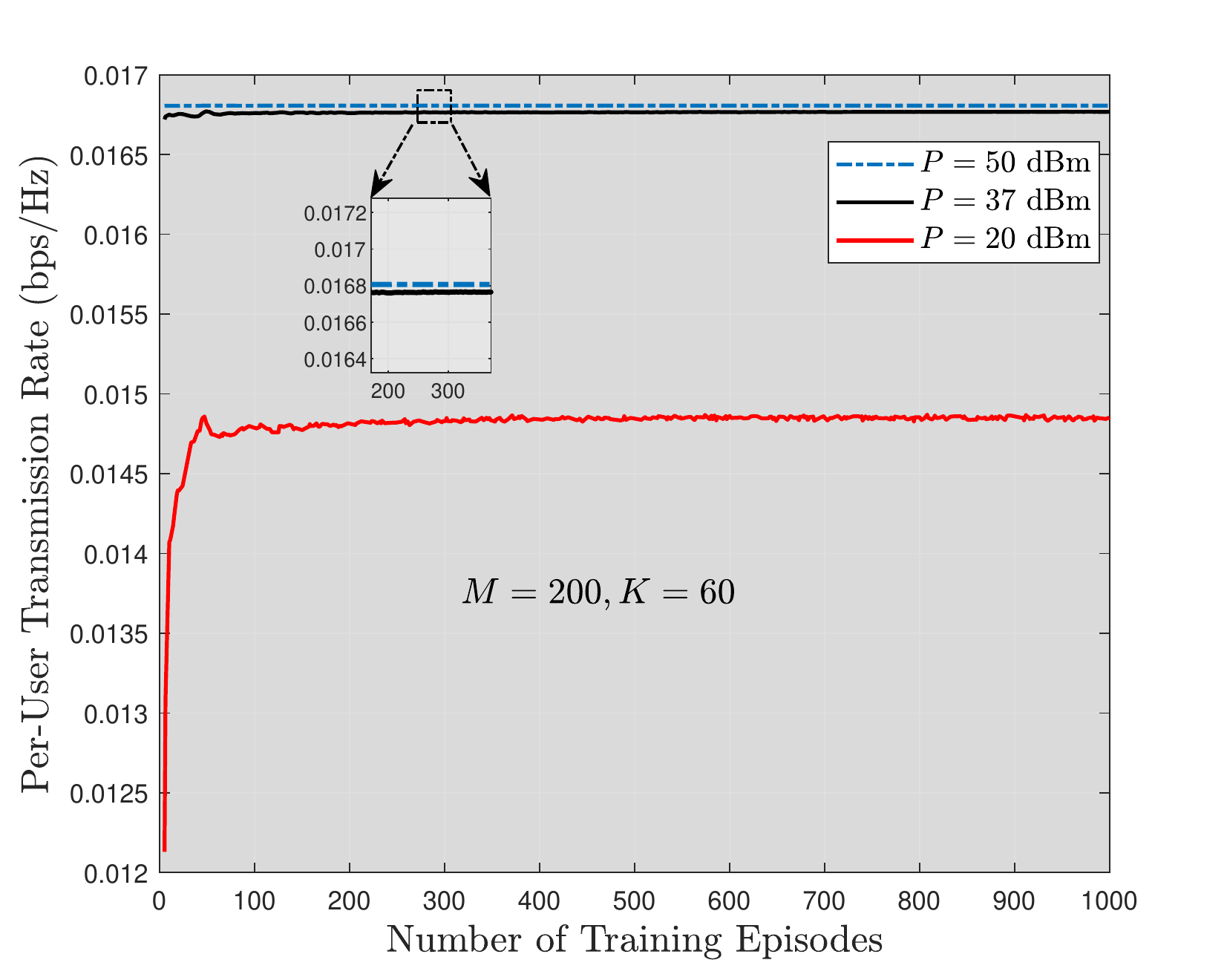}
		\caption{Per-user transmission rate vs. $P$.}\label{DRL_Vs_Power}
\end{figure}
It can be noticed that rate enhancement becomes smaller when the transmission power is increased in a high $P$ regime (e.g. from $P = 37$ dBm to $P = 50$ dBm). 
%This is due to the fact that  transmission power appears at both the numerator and the denominator of the SINR of a user which removes its effect on performance at high values. 
However, a good rate enhancement still can be achieved when $P$ is changed from a relatively small value to a significantly larger value (e.g. from $P = 20$ dBm to $P = 37$ or $P=50$ dBm). Furthermore, the DRL model is found to show a faster convergence with larger values of $P$.  
%==========================
%%%%%%%%%%%%%%%%%%%%%%%%%%%%%%%%%%%%%%%%%%%%%%%%%%%%%%%%%%%%%%%%%%%%%%%%%%%%%%%%%%%%%%%
\section{Conclusion}
 We have first derived closed-form expressions for the  probability of outage for an uplink user in a static  cell-free network.
Next, we have proposed a novel dynamic cell-free network that partitions the distributed APs among a set of subgroups with each subgroup acting as a virtual AP equipped with a distributed antenna system (DAS).
Such a clustering is performed based on the current channel state information (CSI) among the APs and all users within the network coverage area. The dynamic cell-free network model can reduce the complexity for joint signal processing. For the dynamic cell-free model, we have  formulated the general optimization problem of  clustering the APs and designing the beamforming vectors. To solve this optimization problem, we have proposed a hybrid deep deterministic policy gradients (DDPG)-double deep Q-network(DDQN) scheme that jointly selects the optimal network clustering scheme with its optimal beamforming vector values. Possible extensions of this work would include the design and evaluation of more comprehensive DRL models that jointly estimate CSI, select the best clustering configuration for APs, and optimize different beamforming vectors. Also, benchmarking different DRL-based algorithms for optimizing resource allocation in different cell-free network architectures will be valuable.

%This was achieved based on the users' sum-rate optimization scheme for multiple access scenarios. The proposed dynamic cell-free network with hybrid DRL model offered a significant complexity-delay reduction with performance loss satisfactorily compensated by the utilization of SIC detection at the receiver side. Furthermore, sum-rate optimization scheme was found to  significantly overcome the max-min counterpart by better per-user average performance. Finally, the proposed DDPG-DDQN scheme was found to achieve around $78\%$ rate compared to that of exhaustive search. However, complexity and processing time was decreased significantly by the adoption of such DRL model.  
%%%%%%%%%%%%%%%%%%%%%%%%%%%%%%%%%%%%%%%%%%%%%%%%%%%%%%%%%%%%%%%%%%%%%%%%%%%%%%%%%%%%%%%

	\appendices
	\section{}
\setcounter{equation}{0}
First, let us define $X=\sum_{m=1}^M|\tilde{g}_{mk}|^2$ and
\[
Y =    \sum_{m=1}^M
    \left[
    \sum_{l=1, l\neq k}^K|\tilde{g}_{ml}|^2
    +
    \sum_{\dot{l}=1, \dot{l}\neq k}^K
    \sum_{\ddot{l}=1, \ddot{l}\neq \dot{l}}^K
    |\tilde{g}_{m{\ddot{l}}}|^2
    +
    \sum_{\Breve{l}=1, \Breve{l}\neq k}^K|\tilde{g}_{m{\Breve{l}}}|^2
    \right].
    \]
    Let $W=Y+1\approx Y$. (This assumption is valid for large values of $K$ and $M$.) 
The derivation of closed-form expression for $X$ and $Y$ is very complicated due to the non-equal rate parameters ($\tilde{\beta}_{m,k}$ and $\tilde{\beta}_{m,l}$) which makes the MGF method unusable to derive $f_X(x)$ and $f_Y(y)$. Here, we use an accurate approximation of the sum of independent Gamma rvs with different shape and rate parameters, which is referred to as the \textit{Welch-Satterthwaite} approximation \cite{Satterthwaite1946}. Accordingly, we have $X\thicksim \mathcal{G}\left(\Dot{\alpha}_{mk}, \Dot{\beta}_{mk} \right)$ and 
$Y\thicksim \mathcal{G}\left(\Dot{\alpha}_{mk'}, \Dot{\beta}_{mk'} \right)$, 
where 
\[\Dot{\alpha}_{mk}=\frac{\left(\sum_{m=1}^M\tilde{\alpha}_{mk}/ \tilde{\beta}_{mk}\right)^2}{\sum_{m=1}^M\tilde{\alpha}_{mk}/\tilde{\beta}_{mk}^2},~~~~
\Dot{\beta}_{mk}=\frac{\sum_{m=1}^M\tilde{\alpha}_{mk}/\tilde{\beta}_{mk}^2}{\sum_{m=1}^M\tilde{\alpha}_{mk}/ \tilde{\beta}_{mk}},
\]
\[\Dot{\alpha}_{mk'}=\frac{
\left(
\sum_{m=1}^M
    \left[
    \sum_{l=1, l\neq k}^K
    \frac{    \tilde{\alpha}_{ml}}{ \tilde{\beta}_{ml}}
    +
    \sum_{\dot{l}=1, \dot{l}\neq k}^K
    \sum_{\ddot{l}=1, \ddot{l}\neq \dot{l}}^K
    \frac{    \tilde{\alpha}_{m\ddot{l}}}{ \tilde{\beta}_{m\ddot{l}}}
    +
    \sum_{\Breve{l}=1, \Breve{l}\neq k}^K
    \frac{\tilde{\alpha}_{m\Breve{l}}}{        \tilde{\beta}_{m\Breve{l}}}
    \right]
\right)^2
}
{
\sum_{m=1}^M
    \left[
    \sum_{l=1, l\neq k}^K
    \frac{\tilde{\alpha}_{ml}}{    \tilde{\beta}_{ml}^2}
    +
    \sum_{\dot{l}=1, \dot{l}\neq k}^K
    \sum_{\ddot{l}=1, \ddot{l}\neq \dot{l}}^K
    \frac{    \tilde{\alpha}_{m\ddot{l}}}{\tilde{\beta}_{m\ddot{l}}^2}
    +
    \sum_{\Breve{l}=1, \Breve{l}\neq k}^K
    \frac{\tilde{\alpha}_{m\Breve{l}}}{    \tilde{\beta}_{m\Breve{l}}^2}
    \right]
},
\]
\[\Dot{\beta}_{mk'}=
\frac{
\sum_{m=1}^M
    \left[
    \sum_{l=1, l\neq k}^K
    \frac{\tilde{\alpha}_{ml}}{\tilde{\beta}_{ml}^2}
    +
    \sum_{\dot{l}=1, \dot{l}\neq k}^K
    \sum_{\ddot{l}=1, \ddot{l}\neq \dot{l}}^K
    \frac{    \tilde{\alpha}_{m\ddot{l}}}{\tilde{\beta}_{m\ddot{l}}^2}
    +
    \sum_{\Breve{l}=1, \Breve{l}\neq k}^K
    \frac{\tilde{\alpha}_{m\Breve{l}}}{    \tilde{\beta}_{m\Breve{l}}^2}
    \right]
}
{
 \sum_{m=1}^M
    \left[
    \sum_{l=1, l\neq k}^K
    \frac{    \tilde{\alpha}_{ml
}}{\tilde{\beta}_{ml}}
    +
    \sum_{\dot{l}=1, \dot{l}\neq k}^K
    \sum_{\ddot{l}=1, \ddot{l}\neq \dot{l}}^K
    \frac{\tilde{\alpha}_{m\ddot{l}
}}{    \tilde{\beta}_{m\ddot{l}}}
    +
    \sum_{\Breve{l}=1, \Breve{l}\neq k}^K
    \frac{    \tilde{\alpha}_{m\Breve{l}
}}{\tilde{\beta}_{m\Breve{l}}}
    \right]
}.
\]
Note that this approximation becomes identical to the exact expression even for small to moderate number of APs and/or UEs. This makes it useful for modeling large-scale cell-free architectures. 
%The PDF of $W$ can be found through a simple rv transformation as $f_W(w)=f_Y(w-1), w\geq 1$. 
The PDF of the ratio $Z=\frac{x}{w}$ can be defined as
\begin{equation}
    f_Z(z)=\frac{d}{dz}\int_{0}^{\infty}\text{P}\left(x\leq z  w\right)f_Y(w)dw=\int_{0}^{\infty}wf_X(zw)f_Y(w)dw.\label{PDF_Integral_1}
\end{equation}
 Substituting $f_X(zw)$ and $f_Y(w-1)$ into \textbf{A}.1, utilizing \cite[Eq. 3.383.4]{gradshteyn2000}, and simplifying, we obtain  Eq. (\ref{PDF_Centralized}).
%=========================================================
\section{}
\setcounter{equation}{0}
The CDF of the ratio distribution ($Z=\frac{X}{W}$) can be expressed as
\begin{equation}
    F_Z(z)=\text{P}\left(x\leq z W\right) =\int_{0}^{\infty}\text{P}\left(x\leq z w\right)f_Y(w)dw.\label{CDF_Integral_1}
\end{equation}
By substituting $f_Y(w)$ and $P\left(x\leq zw \right)=1-\gamma\left(\Dot{\alpha}_{mk},\Dot{\beta}_{mk}zw\right)$, where $\gamma(.)$ is the lower incomplete gamma function \cite[Eq. 6.5.2]{1965}, and utilizing \cite[Eq. 6.455.2]{gradshteyn2000} and \cite[Eq. 15.3.7]{1965}, and simplifying, we obtain Eq. (\ref{e_outage_1}).
%\hspace{3mm}
\bibliographystyle{IEEEtran}
\bibliography{IEEEabrv,yasser}

\end{document}